\documentclass[aps,pra,twocolumn,10pt,letterpaper]{revtex4-1}
\usepackage{amssymb,amsthm,amsmath,amsfonts}
\usepackage{xcolor,graphicx,ulem}
\usepackage[colorlinks=true,urlcolor=blue,citecolor=blue,linkcolor=blue]{hyperref}
\usepackage{mathptmx}

\newtheorem{theorem}{Theorem}

\definecolor{cg}{rgb}{0.0, 0.7, 0.0}
\definecolor{dp}{rgb}{.7, 0.0, 0.3}

\begin{document}
 
\title{Numerical Construction of Multipartite Entanglement Witnesses}

\author{S. Gerke}\email{stefan.gerke@uni-rostock.de}
\affiliation{Arbeitsgruppe Theoretische Quantenoptik, Institut f\"ur Physik, Universit\"at Rostock, D-18051 Rostock, Germany}

\author{W. Vogel}
\affiliation{Arbeitsgruppe Theoretische Quantenoptik, Institut f\"ur Physik, Universit\"at Rostock, D-18051 Rostock, Germany}

\author{J. Sperling}
\affiliation{Clarendon Laboratory, University of Oxford, Parks Road, Oxford OX1 3PU, United Kingdom}

\date{\today}

\begin{abstract}
	Entanglement in multipartite systems is a key resource for quantum information and communication protocols, making its verification in complex systems a necessity.
	Because an exact calculation of arbitrary entanglement probes is impossible, we derive and implement a numerical method to construct multipartite witnesses to uncover entanglement in arbitrary systems.
	Our technique is based on a substantial generalization of the power iteration---an essential tool for computing eigenvalues---and it is a solver for the separability eigenvalue equations, enabling the general formulation of optimal entanglement witnesses.
	Beyond our rigorous derivation and direct implementation of this method, we apply our approach to several examples of complexly quantum-correlated states and benchmark its general performance.
	Consequently, we provide a generally applicable numerical tool for the identification of multipartite entanglement.
\end{abstract}

\maketitle

\section{Introduction}

	Quantum entanglement is one of the most fundamental concepts in physics.
	It was introduced in the pioneering works of Einstein \textit{et al.} \cite{EPR35} and Schr\"odinger \cite{SCH35}.
	The pure existence of this quantum phenomenon challenged previously established notions of correlations and paved the way towards a new interpretation of the nature of physics.
	Eventually, this led to new protocols used in quantum computing and communication, which utilize the resources of entangled quantum states \cite{NC00}.
	Examples of such classically infeasible tasks are quantum teleportation \cite{BBCJPW93} and dense coding \cite{BW92}.
	Other early protocols concern quantum key distribution, known as BB84 \cite{BB84} and E91 \cite{Ekert91}, and significantly improve communication security.
	Therefore, entanglement plays a key role in fundamental physics and technology-oriented applications.

	A primary concern in the research of entanglement is the actual detection of this quantum correlation.
	Since a lot of protocols for quantum technologies rely on the presence of entanglement, the question whether or not an experimentally generated state is entangled has become a highly relevant topic.
	However, determining entanglement of general states---likewise its counterpart, separability---is an NP-hard problem \cite{G03,I06}.

	Another challenge specific to multipartite systems is the possibility that classical and quantum correlations can be differently distributed among the parties of an ensemble of systems.
	This leads to complex structures of multipartite entanglement; see, e.g., Refs. \cite{HV13,LM13,SSV14}.
	Most notably, there are inequivalent forms of entanglement, which need to be distinguished.
	These are already present in systems of only three qubits, such as the prominent GHZ and W states \cite{DVC00}.
	Beyond that, current experiments become more and more capable of producing large-scale entanglement \cite{YUASKSYYMF13,CMP14,CRFAXFT17}.
	However, while entanglement is vital for characterizing such experiments, the tools to uncover highly quantum-correlated systems are rather limited, and the general verification remains an open problem.

	Still, several criteria have been developed to successfully determine entanglement in bipartite and multipartite systems; see Refs. \cite{T02,HHHH09,GT09} for thorough lists of these entanglement tests.
	A prominent example is the partial transposition criterion \cite{P96}, which has been generalized to general positive, but not completely positive maps \cite{HHH96}.
	Furthermore, such maps are equivalent to entanglement witnesses \cite{HHH96,HHH01,T00,T05}.
	A crucial point of using witness operators is their nature of being observables, which can be directly implemented in experiments.
	Another main advantage is that no full quantum state reconstruction is required to apply such witnesses.
	Rather, a few measurements of the observable can be sufficient to experimentally uncover entanglement \cite{GHBELMS03,BEKGWGHBLS04,TG05}.

	Consequently, witnesses have become a widely applied method for detecting entanglement.
	Their usefulness for quantum technologies has been shown to be promising by detecting entanglement of multipartite cluster states in theory and experiments; see, e.g., Refs. \cite{JMG11,KSWGTUW05}.
	Also, witnesses are not limited to specific systems; for example, they apply to trapped ions \cite{HHRBCCKRRSBGDB05} as well as hybrid systems which correlate vastly different degrees of freedom \cite{BRMM14}.
	In addition, device-independent witnesses have been proposed for a robust verification of entanglement \cite{BRLG13}.
	For instance, such device-independent witnesses can be constructed via so-called matrix-product extensions \cite{RH14}.

	An entanglement witness has a non-negative expectation value for separable states as it defines a hyperplane bisecting the set of states---one part containing at least all separable states and another part including exclusively entangled ones.
	In order to maximize the detectable range of entangled states, optimal witnesses have been introduced \cite{LKCH00,BCHHKLS02,HHMS15,WXCS15,SRLR17}.
	A universally applicable approach is the method of separability eigenvalue equations (SEEs) which enables the construction of optimal witnesses in the bipartite and multipartite scenarios \cite{SV09,SV13}.
	The solution of the SEEs renders it possible, in principle, to formulate all entanglement witnesses.
	However, because of the general complexity of the separability problem, exact solutions are only known for specific scenarios.
	Still, this has already led to deeper insights into the complex forms of experimentally generated multipartite entanglement \cite{GSVCRTF15,GSVCRTF16}.

	Once a witness-construction approach is realized, it can be applied to different physical systems and reveal more insight than the basic indication of entanglement.
	For example, entanglement in systems of indistinguishable particles can significantly differ from the case of distinguishable particles, but witnessing can be done in a similar manner \cite{ESBL02,OK13,RSV15}.
	Furthermore, the quantification of entanglement can be based on witnesses as well \cite{B05,BV06,AP06,EBA07}.
	This also includes entanglement tests for the so-called Schmidt number in the bipartite systems \cite{TH00,SBL01,SV11}, as well as its multipartite extension \cite{EB01,SSV14}.

	Since calculating witnesses is a hard problem and exact solutions are rare, a numerical approach is favorable.
	Numerical methods often use the convexity of the set of separable states.
	Prime examples are approaches based on semidefinite programming, used for the general, convex optimization of linear problems \cite{VB96}.
	The formulation of witnesses has the structure of exactly that kind of problem.
	Thus, semidefinite programming is a frequently applied method for probing entanglement \cite{R01,AM02,PDS02,E03,BV04,DPS05}.
	However, this approach addresses a general class of optimization tasks and is not specifically designed to address the properties of entangled systems.
	Consequently, such a general approach cannot present an optimal strategy to construct entanglement witnesses for arbitrary systems.
	Moreover, numerical standard approaches to solve the eigenvalue equations (EE), such as the well-known power iteration (PI) \cite{GvL13}, do not apply to the construction of entanglement witnesses via the nonlinear SEEs.

	In this contribution, we devise a numerical approach to construct multipartite entanglement witnesses by finding the maximal separability eigenvalue.
	Based on the properties of the SEEs, the analytical background is derived for our technique---termed the \textit{separability power iteration} (SPI).
	As a special case, our approach includes the PI, which returns the maximal solution of EEs.
	We implement the SPI algorithm numerically.
	This is used to demonstrate that the directed design of our numerical approach is an efficient method compared to standard techniques applicable to arbitrary optimization problems.
	To outline possible applications, we use our algorithm, for example, to verify entanglement of weakly correlated, i.e., bound-entangled, states in the bipartite and multipartite scenarios.
	Therefore, an accessible algorithm is provided which renders it possible to construct entanglement probes for certifying multipartite quantum correlations.

	We organize the paper as follows.
	Preliminary statements are made in Sec. \ref{Sec:Preliminaries}.
	Here, we introduce the framework used throughout the contribution and recollect information about entanglement.
	In Sec. \ref{Sec:SPI}, the SPI algorithm to find the maximal separability eigenvalue of a positive operator is introduced.
	Proofs for the working behavior and the convergence of the algorithm are given.
	We analyze the performance of our algorithm in Sec. \ref{Sec:Benchmark}.
	In Sec. \ref{Sec:Example}, entanglement in a selection of bound-entangled states is analyzed.
	In Sec. \ref{Sec:discussion}, we discuss the connection between the SPI and experimental measurements as well as other entanglement criteria and show the broad applicability of our newly devised method to different problems.
	We conclude in Sec. \ref{Sec:Conclusion}, where we also summarize our results.

\section{Preliminaries}\label{Sec:Preliminaries}

	In this section, we revisit multipartite entanglement and its verification.
	In particular, we concentrate on the previously introduced method of SEEs and its relation to standard EEs, which is essential for the following investigations.
	Eventually, we summarize these methods in the context of the considered problem which is solved by our numerical approach, the SPI.

\subsection{Multipartite entanglement}\label{subsec:MultipartiteEntanglement}

	Say $\mathcal{S}$ is the set of all pure states that are separable in an $N$-partite system.
	This means that the elements of $\mathcal S$ take a tensor-product form,
	\begin{align}
		|a_1,\ldots,a_N\rangle=\bigotimes_{j=1,\ldots,N}|a_j\rangle,
	\end{align}
	where $|a_j\rangle\in\mathcal H_j$ is an arbitrary state in the $j$th subsystem and $\langle a_j|a_j\rangle=1$ for $j=1,\ldots,N$.
	Furthermore, a mixed state $\hat{\sigma}$ is separable by definition \cite{W89} if it can be written as
	\begin{align}
		\label{Eq:Separability}
	        \hat{\sigma}=\int dP(a_1,\ldots,a_N)|a_1,\ldots,a_N\rangle\langle a_1,\ldots,a_N|,
	\end{align}
	where $P$ is a classical probability distribution over $\mathcal{S}$.
	Conversely, a state $\hat\rho$ is defined to be entangled if it cannot be expressed in this way.

	The given form of separability is also called full separability of an $N$-partite system.
	To consider instances of partial entanglement, we can assume that each of the $N$ parties is itself a composition of $K_j$ subsystems.
	This allows us to study arbitrary forms of partial separability---e.g., $N$-separability---in a system which, in total, consists of $K_1+\cdots+K_N$ subsystems.
	It is also worth mentioning that continuous-variable entanglement can always be detected in finite-dimensional subspaces \cite{SV09finite}.
	Hence, we can restrict ourselves to Hilbert spaces with a finite dimensionality, $d_j=\dim\mathcal H_j<\infty$.

\subsection{Entanglement witnesses}\label{subsec:EntanglementWitnesses}

	Based on the convexity of the set of separable states [cf. Eq. \eqref{Eq:Separability}], so-called entanglement witnesses, $\hat W$, have been introduced \cite{HHH96,HHH01,T00}.
	They fulfill the property that for all separable states $\hat{\sigma}$, the inequality $\mathrm{tr}(\hat{\sigma}\hat{W})\geq0$ holds true.
	Consequently, entanglement is detected if this inequality is violated, $\mathrm{tr}(\hat{\rho}\hat{W})<0$.
	In particular, it has been shown that witness operators can be written in the form \cite{T05,SV09}
	\begin{align}
		\label{Eq:witnessconstruction}
		\hat W=g_{\max}\hat 1-\hat L,
	\end{align}
	where $g_{\max}$ is the maximal expectation value of $\hat L$ for separable states.

	Therefore, the following approach is equivalent to the method of witnessing \cite{SV09,SV13}:
	For any entangled state $\hat \rho$, there is a Hermitian operator $\hat{L}$ such that the entanglement of $\hat\rho$ is certified by the criterion
	\begin{align}\label{Eq:EntanglementCondition}
		\mathrm{tr}(\hat{L}\hat{\rho})>g_{\max}.
	\end{align}
	The other way around, a state $\hat \sigma$ is separable if for all $\hat L$ the inequality $\mathrm{tr}(\hat{L}\hat{\sigma})\leq g_{\max}$ holds true.
	Moreover, it has been shown that it is sufficient to consider (normalized) positive-definite operators only; see, e.g., Ref. \cite{SV09}.
	We refer to operators satisfying
	\begin{align}\label{Eq:PositiveOperator}
		\hat L=\hat L^\dag>0
	\end{align}
	as positive operators in this work.
	To determine the bound $g_{\max}$, applied in the entanglement criterion \eqref{Eq:EntanglementCondition}, we introduce the SEEs \cite{SV09,SV13} (see also Appendix \ref{App:SEE}).

\subsection{Separability eigenvalue equations}\label{subsection:SeparabilityEigenvalueEquations}

	There are two equivalent forms of the SEEs \cite{SV13}.
	For this work, the more important representation of the SEE reads
	\begin{align}
		\label{Eq:SEE2}
		\hat{L}|a_1,\ldots, a_N\rangle=g|a_1,\ldots, a_N\rangle+|\chi\rangle.
	\end{align}
	Here, the vector $|\chi\rangle$ is $N$ orthogonal to $|a_1,\ldots, a_N\rangle$
	Namely, we have $\langle a_1,\ldots,a_{j-1},x,a_{j+1},\ldots,a_N|\chi\rangle=0$ for all $j=1,\ldots,N$ and for all $|x\rangle\in\mathcal{H}_j$.
	The normalized vector $|a_1,\ldots,a_N\rangle$ is the separability eigenvector.
	The real value $g$ is the separability eigenvalue, which can also be written as the expectation value of $\hat{L}$ with respect to the separability eigenvector,
	\begin{align}
		\label{Eq:gMax}
		g=\langle a_1,\ldots,a_N|\hat L|a_1,\ldots,a_N\rangle.
	\end{align}
	The disturbance to a standard EE, created by the $N$-orthogonal vector $|\chi\rangle$, couples the individual subsystems represented by the states $|a_j\rangle$.
	Thereby, it creates a highly nonlinear equation which, in general, cannot be solved straightforwardly.
	Furthermore, we can relate the separability eigenvalues to our necessary and sufficient entanglement criterion given in inequality \eqref{Eq:EntanglementCondition}.
	Namely, we have \cite{SV13}
	\begin{align}
		g_{\mathrm{max}}=\max\{g: \text{$g$ solves Eq. \eqref{Eq:SEE2}}\}.
	\end{align}

	Let us stress that the maximal separability eigenvalue is the solution to an optimization problem that maximizes the function $\langle a_1,\ldots,a_N|\hat L|a_1,\ldots,a_N\rangle$ for normalized, pure, and separable states.
	Moreover, using relation \eqref{Eq:gMax}, the value of $g_{\max}$ is determined through the corresponding separability eigenvector.
	Finding this vector $|a_1,\ldots,a_N\rangle$ is the goal of our algorithm to be introduced.
	Furthermore, the SEE in Eq. \eqref{Eq:SEE2} takes the form of a perturbed EE.
	In fact, for a single party ($N=1$), the vector $|\chi\rangle$ necessarily vanishes, which means that Eq. \eqref{Eq:SEE2} corresponds to the EE.
	This relation between the SEE and the EE is relevant for our algorithm.

	Furthermore, let us also recall properties of the SEEs, which are of particular importance for this work.
	First, the separability eigenvectors of the operator $\mu\hat L+\nu\hat 1$, for real numbers $\nu$ and $\mu\neq0$, are identical to those of the operator $\hat L$ \cite{SV13}.
	This allows us to restrict ourselves to positive operators, as mentioned above.
	
	The second property to be discussed here addresses the relations between the operators
	\begin{align}
		\hat L=|\xi\rangle\langle\xi|
		\quad\text{and}\quad
		\hat L'=\mathrm{tr}_{N}(\hat L),
	\end{align}
	where $|\xi\rangle\neq0$ is an arbitrary vector in the $N$-partite system and $\mathrm{tr}_{N}$ denotes the partial trace over the $N$th subsystem.
	This also implies that $\hat L'$ is positive semidefinite and acting on an $(N-1)$-partite system.
	The \textit{theorem of cascaded structures} \cite{SV13} states that the nonzero separability eigenvalues of $\hat L$ and $\hat L'$ are identical, which also implies that
	\begin{align}
		g_{\max}=g_{\max}'.
	\end{align}
	Moreover, the separability eigenvectors of $\hat L$ and $\hat L'$ read $|a_1,\ldots,a_N\rangle$ and $|a_1,\ldots,a_{N-1}\rangle$, respectively, where the $N$th component obeys
	\begin{align}
		\label{Eq:NthComponent}
		|a_N\rangle \parallel \langle a_1,\ldots, a_{N-1},\,\cdot\,|\xi\rangle.
	\end{align}
	This means that $|a_N\rangle$ is parallel to a vector that is obtained from $|\xi\rangle$ by projecting its first $N-1$ components onto $\langle a_j|$.
	We emphasize that the optimization of the expectation value of the operator $\hat L$ over $|a_1,\ldots,a_N\rangle\in\mathcal S$, i.e., $|\langle \xi|a_1,\ldots,a_N\rangle|^2\to\max$, corresponds to a maximization using $\hat L'$, which is defined in one subsystem less than used for $\hat L$.
	Also recall that the operator $\hat{L}'$ is, in general, not a rank-1 operator anymore, and the cascaded structure is applicable only to rank-1 operators.

\subsection{Preliminary discussion}

\begin{figure}[b]
	\centering
	\includegraphics[width=.85\columnwidth]{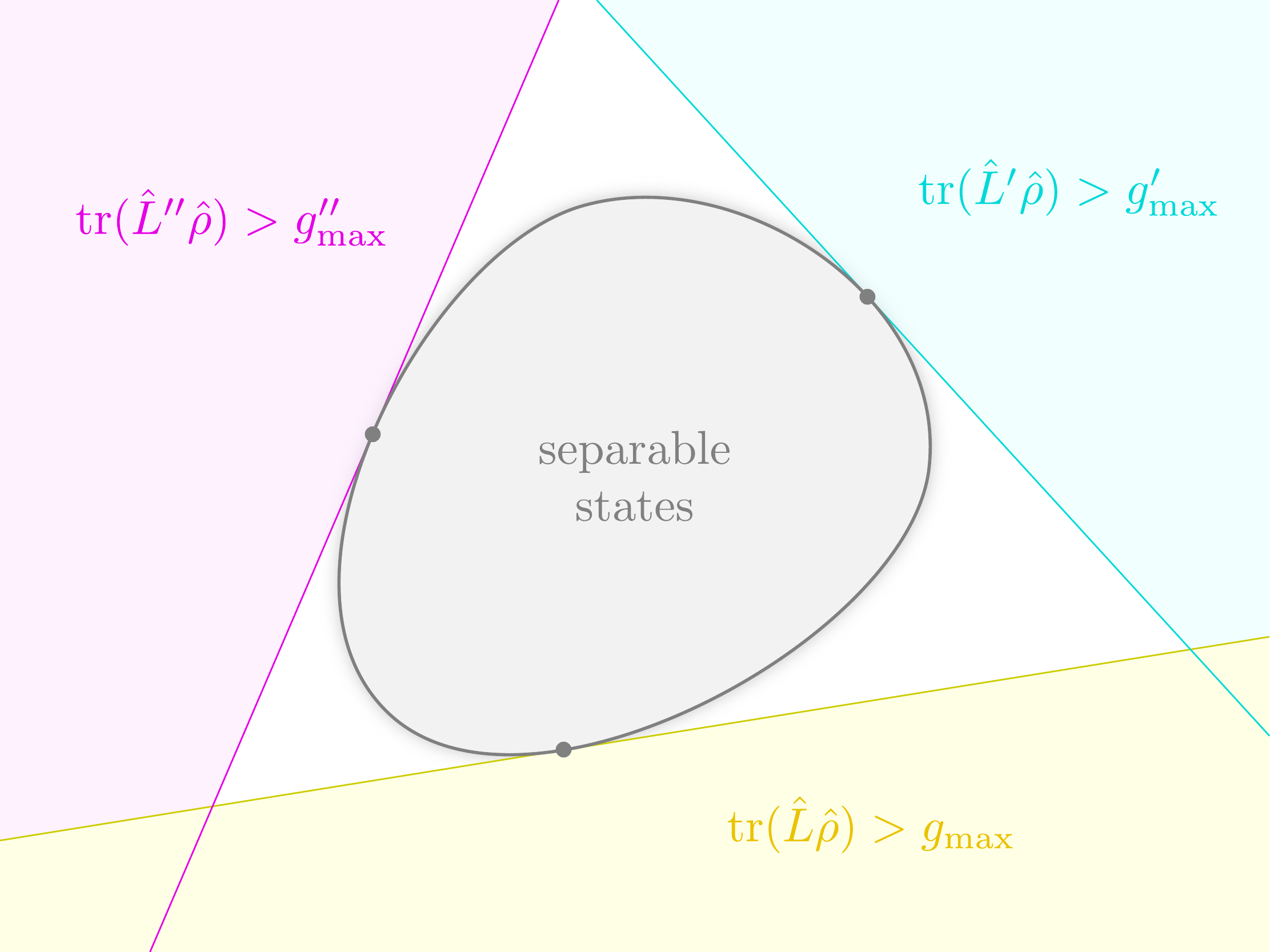}
	\caption{
		Visualization of three entanglement criteria.
		Entanglement is verified in the shaded half-spaces $\mathrm{tr}(\hat L\hat\rho)>g_{\max}$ (bottom, yellow area), $\mathrm{tr}(\hat L'\hat\rho)>g'_{\max}$ (right, cyan area), and $\mathrm{tr}(\hat L''\hat\rho)>g''_{\max}$ (left, magenta area), where the values $g_{\max}$, $g'_{\max}$, and $g''_{\max}$ are the maximal separability eigenvalues of the operators $\hat{L}$, $\hat{L}'$, and $\hat{L}''$ respectively.
		The boundaries define hyperplanes tangent to the set of separable states (gray area).
		The bullet points correspond to the separability eigenvector to the maximal separability eigenvalue for each operator.
	}\label{Fig:Witness}
\end{figure}

	In Fig. \ref{Fig:Witness}, we outline the previously discussed entanglement detection method using three different operators, labeled as $\hat L$, $\hat L'$, and $\hat L''$.
	The tangent hyperplanes separate the set of separable states from states that are verified to be entangled.
	The touching points of the tangent represent the separability eigenvectors to the maximal separability eigenvalue.
	In general, the more operators are used, the better the hyperplanes can approximate the bounds of the set of separable states and the more entangled states can be identified.
	Note that one can construct a dense set of operators for such an approximation with arbitrarily high precision; see, e.g., Ref. \cite{SV09}.

	Both the construction of multipartite entanglement witnesses and the approximation of the set of separable states depend on the solution of the SEEs.
	Specifically, we need to find the maximal separability eigenvalue, which is determined through its corresponding separability eigenvector.
	However, the SEEs present a sophisticated mathematical problem, which has at least the complexity of the standard eigenvalue problem \cite{comment:complexity}.
	In fact, independently of our specific approach, the separability problem has been shown to be an NP-hard problem \cite{G03,I06}.

	Furthermore, the SEE in Eq. \eqref{Eq:SEE2} shares a number of properties with the EE, $\hat L|z\rangle=g|z\rangle$.
	For the latter EE, there exists an algorithm to compute the eigenvector to the maximal eigenvalue of any positive operator $\hat L$, the PI \cite{GvL13}.
	In this algorithm, a vector $|z\rangle$ is mapped onto a new normalized vector, $|z'\rangle=\hat L|z\rangle/\langle z|\hat L^2|z\rangle^{1/2}$.
	An $s$-step iteration, $|z\rangle$, $|z'\rangle$, $|z''\rangle$, $\ldots$, $|z^{(s)}\rangle$, yields a vector that approaches, for $s\to\infty$, an eigenvector to the maximal eigenvalue of $\hat L$ for any initial vector that is not already an eigenvector to $\hat L$.

	In the following, we aim to generalize the PI to be applicable to the SEE.
	For this reason, we introduce an algorithm for a numerical implementation, which yields the desired solution of the SEEs---a separability eigenvector to the maximal separability eigenvalue.
	The resulting SPI algorithm is applicable to all positive operators $\hat L$ and enables the construction of witnesses to probe multipartite entanglement.

\section{The SPI algorithm}\label{Sec:SPI}

	In this section, we present the SPI algorithm---step by step.
	The flowchart of this algorithm to construct entanglement criteria is shown in Fig. \ref{Fig:Algorithm}.
	Our approach yields the separability eigenvector $|a_1,\ldots,a_N\rangle$ to the desired, maximal separability eigenvalue for a positive operator $\hat L$ [Eq. \eqref{Eq:gMax}].
	Before we study the individual, essential parts of the SPI in a rigorous mathematical framework, let us first get a general overview of how our algorithm operates by applying it to an example.

%!!! FALSE POSITIONING IN SOURCE FILE FOR PROPER POSITION IN OUTPUT
\begin{figure*}
	\centering
	\includegraphics[width=.85\textwidth]{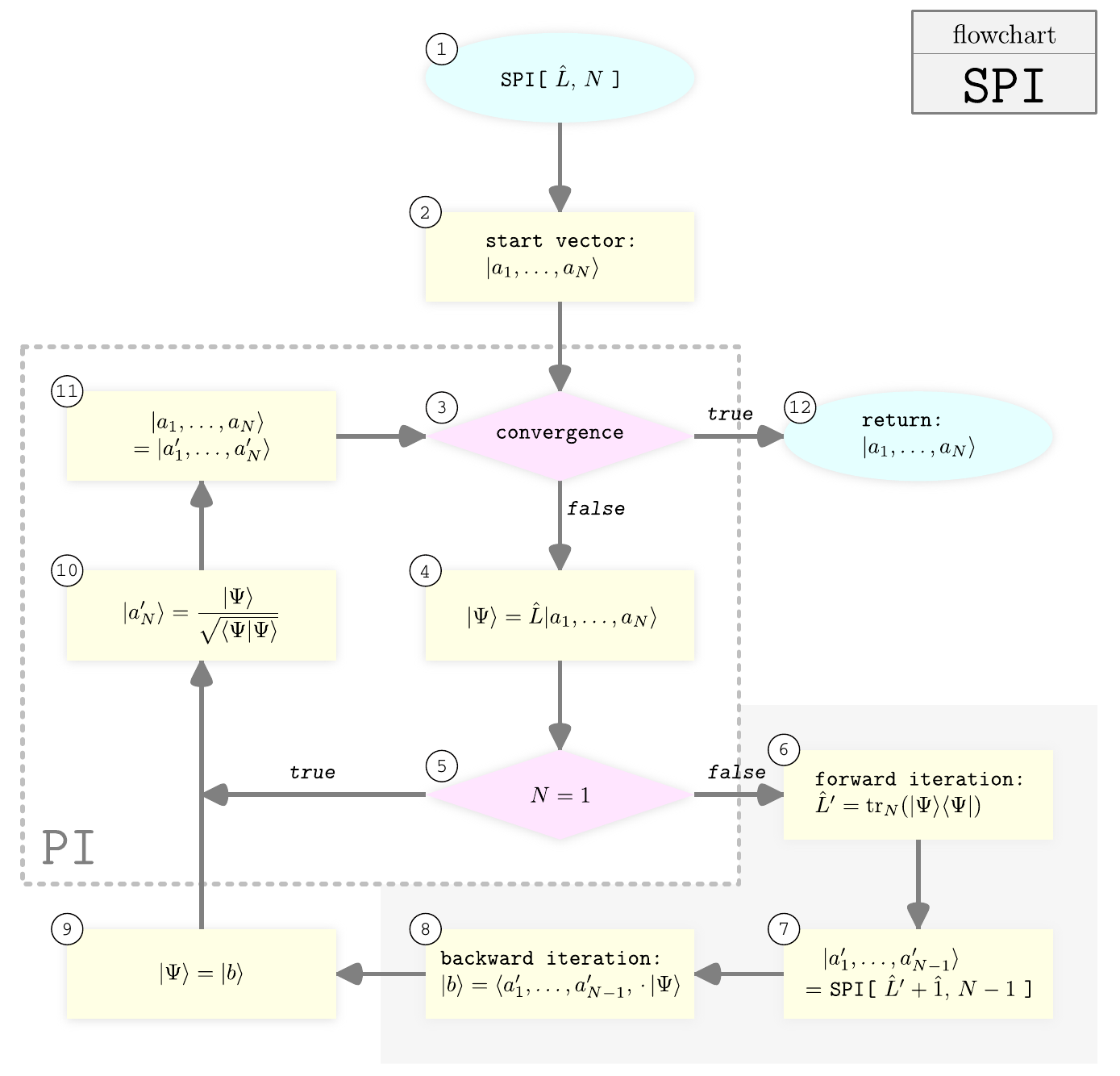}
	\caption{
		Flowchart of the SPI algorithm.
		Branches in the algorithm, either ``while'' loops or ``if'' conditions, are represented by magenta diamonds.
		Yellow rectangles represent assignments and function calls.
		Entry and exit points of the algorithm are shown as cyan ellipses.
		Box \textcircled{\scriptsize 1} refers to the input, which includes an operator and the number $N$ of parties.
		A vector is generated in \textcircled{\scriptsize 2} to serve as our starting vector.
		If the convergence criterion \textcircled{\scriptsize 3}, studied in Sec. \ref{subsec:Convergence}, is not met, we generate another $N$-partite vector in \textcircled{\scriptsize 4}.
		For $N=1$ in \textcircled{\scriptsize 5}, the algorithm corresponds to the power iteration (PI) which finds the standard eigenvector to the maximal eigenvalue (dashed box).
		For $N>1$, our extension consists of three essential parts (gray areas), which are the forward iteration \textcircled{\scriptsize 6} and the backward iteration \textcircled{\scriptsize 8} as well as a recursion calling the SPI for $N-1$ parties in \textcircled{\scriptsize 7}.
		Eventually, when \textcircled{\scriptsize 3} is satisfied, the desired separability eigenvector is the output of our SPI algorithm and returned in \textcircled{\scriptsize 12}.
	}\label{Fig:Algorithm}
\end{figure*}
%!!! FALSE POSITIONING IN SOURCE FILE FOR PROPER POSITION IN OUTPUT

\subsection{Proof of concept}\label{Sec:Concept}

	For demonstrating the function of our algorithm, we consider the bipartite ($N=2$) and positive operator
	\begin{align}
		\label{Eq:ExampleRunOperator}
		\hat L=2\hat 1-\hat V,
	\end{align}
	where $\hat V$ is the swap operator, $\hat V|a_1,a_2\rangle=|a_2,a_1\rangle$.
	The expectation value $\langle a_1,a_2|\hat L|a_1,a_2\rangle=2-|\langle a_1|a_2\rangle|^2$ directly implies that the maximal separability eigenvalue is $g_{\max}=2$, and it is attained for $|a_1\rangle\perp|a_2\rangle$ \cite{SV09}.
	This exact result serves as our reference to assess the success of our algorithm for this example.
	Moreover, since the maximal standard eigenvalue is three, it follows from $g_{\max}<3$ that this operator can be used to detect entanglement \cite{SV11}.
	In fact, the swap operator is related to the prominent partial transposition criterion to verify entanglement \cite{P96,HHH96,SV09}.

	Our algorithm in Fig. \ref{Fig:Algorithm} is initialized at point \textcircled{\scriptsize 1} with the operator \eqref{Eq:ExampleRunOperator} and the number of subsystems being $N=2$.
	At \textcircled{\scriptsize 2}, let us begin with states $|a_1,a_2\rangle$, which are neither parallel nor orthogonal, to exclude the trivial cases.
	Namely, we have $0< |\gamma|^2<1$, where
	\begin{align}
		\label{Eq:ExampleRunGamma}
		\gamma=\langle a_1|a_2\rangle.
	\end{align}
	Say that in step \textcircled{\scriptsize 3}, we do not have convergence yet; i.e., we follow the branch labeled ``false'' and compute the vector in step \textcircled{\scriptsize 4},
	\begin{align}
		|\Psi\rangle=\hat L|a_1,a_2\rangle
		=2|a_1,a_2\rangle-|a_2,a_1\rangle.
	\end{align}
	Since $N\neq 1$ (step \textcircled{\scriptsize 5}), we proceed to \textcircled{\scriptsize 6} and compute the operator,
	\begin{align}
	\begin{aligned}
		& \hat L'=\mathrm{tr}_2|\Psi\rangle\langle\Psi|
		\\
		= & 4|a_1\rangle\langle a_1|-2\gamma|a_1\rangle\langle a_2|
		-2\gamma^\ast |a_2\rangle\langle a_1|+|a_2\rangle\langle a_2|,
	\end{aligned}
	\end{align}
	which is a single-subsystem operator.
	This step is referred to as \textit{forward iteration} in Fig. \ref{Fig:Algorithm}.
	The idea behind this step is a result of the theorem of cascaded structures, which finds the maximal separable projection onto the state $|\Psi\rangle$; see Sec. \ref{subsection:SeparabilityEigenvalueEquations}.
	This also allows us to apply the SPI to $\hat L'+\hat 1$.
	As $\hat L'$ is positive semidefinite by construction, the addition of $\hat 1$ assures the positivity of $\hat L'+\hat 1$ without modifying the separability eigenvectors.

	Calling the SPI with $N\mapsto N-1=1$ in \textcircled{\scriptsize 7} and thereby going back to step \textcircled{\scriptsize 1} and going through steps \textcircled{\scriptsize 2} to \textcircled{\scriptsize 5}, we follow the branch for which $N=1$ is true.
	This gives an iteration of steps \textcircled{\scriptsize 10}, \textcircled{\scriptsize 11}, \textcircled{\scriptsize 3}, \textcircled{\scriptsize 4}, and \textcircled{\scriptsize 5}, indicated through the dashed box in Fig. \ref{Fig:Algorithm}, which describes the PI.
	The PI is employed for  solving the standard EE numerically by returning the eigenvector to the maximal eigenvalue of a positive operator with an arbitrarily high precision.
	So we can assume that the convergence \textcircled{\scriptsize 3} is true after some iterations of the PI.
	For the given operator $\hat L'+\hat 1$ (thus, also for $\hat L'$), the eigenvector to the maximal eigenvalue reads
	\begin{align}
		\label{Eq:ExampleRunState1}
		|a_1'\rangle=\frac{1}{\nu}\left(
			4|\gamma||a_1\rangle
			+
			[\Gamma-3]\frac{\gamma^\ast}{|\gamma|}|a_2\rangle
		\right),
	\end{align}
	using $\gamma$ [cf. Eq. \eqref{Eq:ExampleRunGamma}], the abbreviation
	\begin{align}
		\Gamma=\sqrt{9-8|\gamma|^2},
	\end{align}
	and the normalization $\nu=[2\Gamma(\Gamma-3+4|\gamma|^2)]^{1/2}$.

	Thus, the PI basically returns the vector $|a_1'\rangle$ in step \textcircled{\scriptsize 12}, which is used to continue with the case $N=2$, where we exit step \textcircled{\scriptsize 7} to perform the \textit{backward iteration} step \textcircled{\scriptsize 8}.
	This gives a vector in the second subsystem.
	For convenience, this vector is renamed (step \textcircled{\scriptsize 9}) and normalized (step \textcircled{\scriptsize 10}); see Fig. \ref{Fig:Algorithm}.
	Again, the backward iteration is a result of the theorem of cascaded structures, which relates the separability eigenvectors of $\hat L'$ for $N-1$ subsystems with those of the initial operator $\hat L$ for $N$, cf. Sec. \ref{subsection:SeparabilityEigenvalueEquations}.
	This yields the state of the second subsystem,
	\begin{align}
		\label{Eq:ExampleRunState2}
		|a_2'\rangle=\frac{1}{\nu}\left(
			4|\gamma||a_2\rangle
			+
			[\Gamma-3]\frac{\gamma}{|\gamma|}|a_1\rangle
		\right).
	\end{align}
	Thus, we obtain a new separable state $|a'_1,a'_2\rangle$ in \textcircled{\scriptsize 11}, where the tensor-product state is formed.

	What did we achieve with the construction of this new state?
	To answer this question, let us recall that the desired separability eigenvector of the operator under study has perpendicular components for the subsystems.
	Thus, in analogy to Eq. \eqref{Eq:ExampleRunGamma}, we may compute the scalar product of the states of the subsystems, which yields $\gamma'=\gamma/\Gamma$ and
	\begin{align}
		|\gamma'|^2=|\langle a_1'|a_2'\rangle|^2=\frac{|\gamma|^2}{9-8|\gamma|^2}<|\gamma|^2=|\langle a_1|a_2\rangle|^2.
	\end{align}
	This means the states $|a_1'\rangle$ and $|a_2'\rangle$ are closer to orthogonal than the initial states $|a_1\rangle$ and $|a_2\rangle$.
	Equivalently, we can say that the expectation value of $\hat L$ increases, $\langle a_1',a_2'|\hat L|a_1',a_2'\rangle=2-|\gamma'|^2>\langle a_1,a_2|\hat L|a_1,a_2\rangle=2-|\gamma|^2$.
	Now, we can perform the next cycle, which results in $|\gamma|^2>|\gamma'|^2>|\gamma''|^2$.
	In fact, performing $s$ steps of the SPI, we get vectors $|a_1^{(s)},a_2^{(s)}\rangle$ for which
	\begin{align}
		\left|\langle a_1^{(s)}|a_2^{(s)}\rangle\right|^2
		<\frac{|\gamma|^2}{\left(9-8|\gamma|^2\right)^s}
		\stackrel{s\to\infty}{\longrightarrow}0
	\end{align}
	holds.
	Therefore, we get a convergent sequence of separability eigenvectors which, in the limit of infinite iterations, yields the desired exact maximal separability eigenvalue, $\langle a_1^{(s)},a_2^{(s)}|\hat L|a_1^{(s)},a_2^{(s)}\rangle\to2-0=g_{\max}$ for $s\to\infty$.

	In conclusion of this example resulting in an entanglement test based on the swap operator [cf. Eqs. \eqref{Eq:ExampleRunOperator} and \eqref{Eq:EntanglementCondition}], our SPI is constructed to deliver the separability eigenvector to the maximal separability eigenvalue.
	Applying properties of the theorem of the cascaded structure of SEEs, we identify the following essential steps: forward and backward iteration.
	The forward iteration allows the reduction of the number of subsystems by one in each recursion depth until the recursion depth reaches a maximum when the operator is a single-partite operator.
	Then the SEE reduces to the EE, and the PI is used to get the eigenvector to the maximal eigenvalue.
	The eigenvector is further used in the next step, the backward iteration, to obtain the remaining subsystem components of the separable product vector.
	After completing multiple instances of such a cycle, we obtain an arbitrarily precise approximation to our sought-after separability eigenvector.

	Now, we may consider the general case beyond the specific example, which was used to demonstrate the general operation of our generally applicable algorithm in Fig. \ref{Fig:Algorithm}.
	This gives the mathematically rigorous formulation of the SPI for arbitrary positive operators $\hat L$ and arbitrary numbers $N$ of subsystems, which necessarily requires a rather technical treatment because of the complexity of the underlying separability problem.
	After this, we perform a benchmarking of our algorithm and apply it to various examples, which provides a more intuitive assessment of our method.

\subsection{Analytic framework}

	Based on the theorem on cascaded structures for the SEEs, the SPI iterates over the number of parties from $N$ to one.
	For $N=1$, the SPI and PI are identical, resembling the underlying fact that the SEE and EE are the same in this case too.
	Beyond the PI, the SPI algorithm includes two main steps, denoted as forward and backward iteration.
	Clearly, the major goal of our maximization algorithm for a positive operator $\hat L$ is to get a new separable state $|a'_1,\ldots,a'_N\rangle$ from the preceding state $|a_1,\ldots,a_N\rangle$, which increases the expectation value, $\langle a'_1,\ldots,a'_N|\hat L|a'_1,\ldots,a'_N\rangle>\langle a_1,\ldots,a_N|\hat L|a_1,\ldots,a_N\rangle$.
	Here, let us discuss the details, the proofs of some of the required theorems are provided in the corresponding appendixes.

\subsubsection{Initial considerations}\label{Sec:InitialThoughts}

	Let us make some more general observations, which we then apply to the separability problem under study.
	A positive operator $\hat L$ induces a scalar product,
	\begin{align}
		\langle x|y\rangle_{\hat L}=\langle x|\hat L|y\rangle,
	\end{align}
	for arbitrary $|x\rangle$ and $|y\rangle$.
	Therefore, the Cauchy-Schwarz inequality holds true, $|\langle x|y\rangle_{\hat L}|^2\leq \langle x|x\rangle_{\hat L}\langle y|y\rangle_{\hat L}$, where the equality is equivalent to $|x\rangle\parallel|y\rangle$.
	Also, we have $\langle x|x\rangle_{\hat L}>0$ for all $|x\rangle\neq0$.
	To apply these features, we have to restrict ourselves to positive operators $\hat L$.
	Note that in our following proofs, we rely on the properties of the scalar product; for example, a positive-semidefinite operator $\hat{L}$ would be insufficient \cite{comment:positivity}.

	Say $\mathcal T$ is a closed and bounded subset of a finite-dimensional Hilbert space.
	For a $|z\rangle\in\mathcal T$, one can define an iterated state as
	\begin{align}
		\label{Eq:ArgMax}
		|z'\rangle=\underset{|y\rangle\in\mathcal T}{\mathrm{arg\,max}\,}|\langle y|z\rangle_{\hat L}|,
	\end{align}
	where we use the function ``$\mathrm{arg\,max}$,'' which returns the argument for which the maximum is reached.
	In other words, $|z'\rangle=\mathrm{arg\,max}_{|t\rangle\in\mathcal T}\, |\langle t|z\rangle_{\hat L} |$ if $|z'\rangle\in\mathcal T$ satisfies the relation $|\langle z'|z\rangle_{\hat L} |=\max_{|t\rangle\in\mathcal T} |\langle t|z\rangle_{\hat L}|$.
	Since $|z\rangle$ is also an element of the set $\mathcal T$ over which we maximize, we can conclude that $\langle z|z\rangle_{\hat L}\leq|\langle z'|z\rangle_{\hat L}|$.
	Applying the Cauchy-Schwarz inequality, we get
	\begin{align}
		\langle z|z\rangle_{\hat L}^2
		\leq
		|\langle z'|z\rangle_{\hat L}|^2
		\leq
		\langle z'|z'\rangle_{\hat L}
		\langle z|z\rangle_{\hat L}
		\leq
		\langle z'|z'\rangle_{\hat L}
		|\langle z'|z\rangle_{\hat L}|.
	\end{align}
	Considering the second and fourth terms, as well as the first and third terms, we find the increasing sequence 
	\begin{align}
		\langle z|z\rangle_{\hat L}\leq |\langle z'|z\rangle_{\hat L}|\leq \langle z'|z'\rangle_{\hat L}.
	\end{align}
	From the definition of $|z'\rangle$ and the properties of the Cauchy-Schwarz inequality, we can also conclude that the equality holds true if and only if $|z'\rangle\parallel|z\rangle$.

	Therefore, we can state that the iteration $|z\rangle$, $|z'\rangle$, $|z''\rangle$, etc. produces a sequence of increasing expectation values, $\langle z|\hat L|z\rangle\leq\langle z'|\hat L|z'\rangle\leq\langle z''|\hat L|z''\rangle\leq\cdots$.
	However, the elusive $\mathrm{arg\,max}$ function \eqref{Eq:ArgMax} has to be computed for this purpose.
	In fact, this can be done for separable states, $\mathcal T=\mathcal S$.

	We may use the abbreviation $|\Psi\rangle=\hat L|a_1,\ldots,a_N\rangle$.
	To maximize the projections of this state onto separable ones, we can apply the theorem of the cascaded structure.
	This means that the maximal projection of this state onto separable states is obtained by $| a'_1,\ldots,a'_N\rangle$ for the maximal separability eigenvalue of $|\Psi\rangle\langle\Psi|$.
	In Sec. \ref{subsection:SeparabilityEigenvalueEquations} and in the flowchart of the SPI in Fig. \ref{Fig:Algorithm}, we describe how this is achieved:
	We reduce the number of parties $N$ and solve the SEE for $\hat L'=\mathrm{tr}_N|\Psi\rangle\langle\Psi|$ (forward iteration, \textcircled{\scriptsize 6}) to get $|a'_1,\ldots,a'_{N-1}\rangle$ (step \textcircled{\scriptsize 7}), which then determines the remaining component $|a'_N\rangle$ from $\langle a'_1,\ldots,a'_{N-1},\,\cdot\,|\Psi\rangle$ (backward iteration, \textcircled{\scriptsize 8}).

	In summary, the cascaded structure describes how to compute the desired $\mathrm{arg\,max}$ function for separable states.
	This describes the underlying principle of the SPI, which allows us to compute the bounds $g_{\max}$ for the necessary and sufficient entanglement criteria \eqref{Eq:EntanglementCondition}.

\subsubsection{The SPI}

	To apply the general relations above, let us begin with the forward iteration step.
	By the following Theorem \ref{Th:Forward}, it is guaranteed that finding the separability eigenvector corresponds to determining the maximal separability eigenvalue for the $(N-1)$-partite case.
	More specifically, it enables us to reduce the number of subsystems for the SEE by one.
	\begin{theorem}[Forward iteration]
		Let $|a_1,\ldots,a_N\rangle$ be the separability eigenvector corresponding to the maximal separability eigenvalue of a positive $N$-partite operator $\hat{L}$.
		Furthermore, let $|\Psi\rangle=\hat{L}|a_1,\ldots,a_N\rangle$.
		For the $(N-1)$-partite operator $\hat{L}'=\mathrm{tr}_N(|\Psi\rangle\langle\Psi|)$, the equality
		\begin{align}
			\label{Eq:THM1}
		\begin{aligned}
			&\langle a_1,\ldots,a_N|\hat{L}|a_1,\ldots,a_N\rangle
			\\
			=&\sqrt{\langle a_1,\ldots,a_{N-1} |\hat{L}'|a_1,\ldots,a_{N-1}\rangle}
		\end{aligned}
		\end{align}
		holds true.
		See Appendix \ref{App:Forward} for the proof.
		\label{Th:Forward}
	\end{theorem}
	This theorem is a direct consequence of the SEE in Eq. \eqref{Eq:SEE2} and its properties.
	In the SPI algorithm, the theorem is applied in the forward iteration step \textcircled{\scriptsize 6}.
	To find the full $N$-partite separability eigenvector, a reverse step has to be taken.
	Theorem \ref{Th:Backward} states how the $N$ subsystem separability eigenvector can be generated from the $N-1$ subsystem separability eigenvector.
	\begin{theorem}[Backward iteration]
		Consider the same definitions used in Theorem \ref{Th:Forward}.
		If the $(N-1)$-partite separability eigenvector $|a_1,\ldots,a_{N-1}\rangle$ maximizes Eq. \eqref{Eq:THM1}, then the separability eigenvector $|a_1,\ldots,a_{N-1}\rangle\otimes|a_N\rangle$ maximizes $\langle a_1,\ldots,a_{N}|\hat{L}|a_1,\ldots,a_{N}\rangle$ if the condition
		\begin{align}
			\nu|a_N\rangle=\langle a_1,\ldots,a_{N-1},\,\cdot\,|\Psi\rangle
			\label{Eq:THM2}
		\end{align}
		holds true for $\nu\in\mathbb{C}\setminus\{0\}$.
		\label{Th:Backward}
	\end{theorem}
	The proof of this theorem directly follows from the cascaded structure, cf. Eq. \eqref{Eq:NthComponent}.
	It relates the separability eigenvector for $N$ parties to those of a lower number of parties, $N-1$.
	Thereby, if a solution to the $(N-1)$-partite SEEs for $\hat L'$ is known, we directly find the $N$th component of the full solution $|a_1,\ldots,a_N\rangle$.
	In the flowchart in Fig. \ref{Fig:Algorithm}, we see the application in the backward iteration step \textcircled{\scriptsize 8}.

	The combination of Theorems \ref{Th:Forward} and \ref{Th:Backward} is fundamental for the SPI to work.
	In fact, one might visualize the working principle of the algorithm as a nested cascading structure.
	The forward iteration is recursively applied until we reach the case $N=1$.
	In that case, the standard PI is performed.
	After that, the backward iteration finalizes the individual recursion layers of the SPI until we obtain the new $N$-partite separable vector.
	Then, we can start a new cycle of forward iterations, the PI, and backward iterations until the convergence is reached.
	The algorithm will terminate successfully and return the complete separability eigenvector corresponding to the maximal separability eigenvalue $g_{\max}$ for detecting entanglement in terms of inequality \eqref{Eq:EntanglementCondition}.

	To verify the statement that the algorithm converges to the maximal separability eigenvalue, a few observations have to be shown first.
	Let us take a closer look at the sequence of product vectors created by the SPI.
	In every step, we find an element of all product states, $|a'_1,\ldots,a'_N\rangle\in\mathcal S$, which projects maximally onto the action of operator $\hat{L}$ onto the previously generated product state $|a_1,\ldots,a_N\rangle$.
	This iteration is done until we reach convergence.
	This generates a monotonously growing sequence of expectation values of $\hat{L}$, which is stated in the following theorem.
	\begin{theorem}[Monotony]
		Let $|a_1,\ldots,a_N\rangle\in\mathcal S$ and $|a'_1,\ldots,a'_N\rangle\in\mathcal{S}$ such that
	        \begin{align}
	                |a'_1,\ldots,a'_N\rangle=\underset{|b_1,\ldots,b_N\rangle\in\mathcal S}{\mathrm{arg\,max}}\langle b_1,\ldots,b_N|\hat{L}|a_1,\ldots,a_N\rangle.
	        \end{align}
		Then the inequality
		\begin{align}
			\label{Eq:Monotony}
		\begin{aligned}
			&\langle a_1,\ldots,a_N|\hat{L}|a_1,\ldots,a_N\rangle
			\\
			\leq&
			\langle a'_1,\ldots,a'_N|\hat{L}|a_1,\ldots,a_N\rangle
			\\
			\leq&
			\langle a'_1,\ldots,a'_N|\hat{L}|a'_1,\ldots,a'_N\rangle
		\end{aligned}
		\end{align}
		holds true.
		Furthermore, equality in Eq. \eqref{Eq:Monotony} holds true iff $|a'_1,\ldots,a'_N\rangle=| a_1,\ldots,a_N\rangle$.
		See Appendix \ref{App:Monotony} for the proof.
		\label{Th:Monotony}
	\end{theorem}
	This theorem is a special case of the general considerations made in Sec. \ref{Sec:InitialThoughts}.
	In addition, the global phase of the separable state $|a'_1,\ldots,a'_N\rangle$ can be chosen freely, which we conveniently select such that we have positive projections onto $|\Psi\rangle$, i.e., $\langle a'_1,\ldots,a'_N|\hat{L}|a_1,\ldots,a_N\rangle=|\langle a'_1,\ldots,a'_N|\hat{L}|a_1,\ldots,a_N\rangle|$.
	Let us stress that the $\mathrm{arg\,max}$ function, i.e., finding the maximal projection onto $|\Psi\rangle=\hat L|a_1,\ldots,a_N\rangle$, is obtained from the cascaded structure; see also Theorems \ref{Th:Forward} and \ref{Th:Backward}.

	Because of Theorem \ref{Th:Monotony}, the SPI produces a sequence of increasing expectation values of $\hat L$.
	This observation is an important aspect for the proof of convergence of the SPI, which is shown in two parts.
	Both theorems rely on the sequence $(g^{(s)})_s$ of expectation values generated by the SPI in each step $s$, where
	\begin{align}
		\label{Eq:SEValueSequence}
		g^{(s)}=\langle a_1^{(s)},\ldots,a_N^{(s)}|\hat{L}|a_1^{(s)},\ldots,a_N^{(s)}\rangle.
	\end{align}
	Here, in analogy to the example in Sec. \ref{Sec:Concept}, the vector $|a_1^{(s)},\ldots,a_N^{(s)}\rangle$ is the approximation to the separability eigenvector for the maximal separability eigenvalue after $s$ iterations of the SPI.
	First, we consider the local convergence of the algorithm.
	\begin{theorem}[Local convergence]
	        For any starting vector, the sequence $(g^{(s)})_s$ of expectation values generated by the SPI converges, i.e., the limit
	        \begin{align}
			\label{Eq:LocalConvergence}
	                \lim_{s\to\infty} g^{(s)} = \bar{g}
	        \end{align}
		exists and is bounded as $0\leq\bar{g}\leq g_{\mathrm{max}}$.
		See Appendix \ref{App:Local} for the proof.
	        \label{Th:Local}
	\end{theorem}
	For an arbitrary starting vector, a sequence of expectation values of $\hat{L}$ for separable states is generated.
	The generated sequence converges independent of the choice of starting vector.
	Combining the statements from Theorems \ref{Th:Monotony} and \ref{Th:Local}, we conclude that there is a monotone growth of expectation values towards a maximum.
	This maximum does not necessarily need to be the maximal separability eigenvalue $g_{\mathrm{max}}$ as shown in Theorem \ref{Th:Local}.
	We therefore require an additional observation to prove global convergence of the SPI, which is stated in Theorem \ref{Th:Global}.
	\begin{theorem}[Global convergence]
		Let $\Sigma$ be a set of separable starting vectors and $(g_{\Phi}^{(s)})_s$ be sequences of expectation values generated by the SPI for a starting vector $|{\Phi}\rangle\in\Sigma$.
		Furthermore, say $\{\bar{g}_{\Phi}\in\mathbb{R}:|{\Phi}\rangle\in\Sigma \text{ and } \bar{g}_{\Phi}=\lim_{s\to\infty} g_{\Phi}^{(s)}\}$ defines the set of optimal expectation values (limits of the converged sequences) for each starting vector.
		The maximal separability eigenvalue for the operator $\hat{L}$ is $g_\mathrm{max}=\mathrm{max}_{\Phi\in\Sigma}\{\bar{g}_{\Phi}\}$, which is the maximum of the limit to the series of expectation values for each starting vector.
		See Appendix \ref{App:Global} for the proof.
		\label{Th:Global}
	\end{theorem}
	The set $\Sigma$ of different starting vectors $|\Phi\rangle$ that we consider is covered in Sec. \ref{subsec:StartingVectors}.
	Even in the worst-case scenario, it is far smaller than the set $\mathcal S$ of all separable states.

\subsubsection{Starting vectors}\label{subsec:StartingVectors}

	An important aspect for the implementation of the algorithm is the choice of a starting vector, cf. Theorem \ref{Th:Global}.
	Because a proper choice can significantly decrease the runtime of the algorithm, let us provide more details on this aspect.

	Assume we start with a separability eigenvector corresponding to any---except for the largest---separability eigenvalue.
	Then, the algorithm converges immediately, and the resulting separability eigenvalue will not be maximal.
	It is worth mentioning that such a behavior is already well known for the PI.
	It is straightforward to check whether an initial vector is a (separability) eigenvector.
	Similar results might happen for starting vectors that are too close to any (separability) eigenvector.

	To circumvent such problems, the SPI can be run multiple times with different starting vectors, chosen as an operator basis.
	Namely, the set $\{|a_{1,k},\ldots,a_{N,k}\rangle\langle a_{1,k},\ldots,a_{N,k}|\}_{k}$ of starting vectors spans all operators of the underlying Hilbert space.
	This allows us to cover all parts of the operator space and, of course, also resolves the related problem for the PI.

	This choice is valid as the set of separability eigenvectors can be used to find a decomposition of any state, similarly to the spectral decomposition found by regular eigenvectors.
	In fact, any positive-semidefinite operator can be decomposed in terms of projectors of separability eigenvectors; see Refs. \cite{SV09quasi,SW18} for proofs of the bipartite and multipartite cases, respectively.
	Specifically, the decomposition of at least one vector of the operator basis needs to contain the sought-after separability eigenvector.
	This warrants the choice of using the operator basis as starting vectors.

	Another efficient \textit{ad hoc} ansatz that we used for the implementation of the SPI is described as follows:
	First, a preliminary run of the SPI is done to find a product vector projecting maximally onto the vector $\hat{L}|v\rangle$, where $|v\rangle$ is the eigenvector to the maximal standard eigenvalue of $\hat{L}$.
	This choice is inspired by the fact that the wanted vector $|a_1,\ldots,a_N\rangle$ maximizes the expectation value of $\hat{L}$ with respect to product vectors.
	The eigenvector $|v\rangle$ maximizes the expectation value of $\hat{L}$ without the restriction to separable states.
	Second, the product vector that lies maximally parallel to $|v\rangle$ serves as our starting vector.
	Finding such a maximal projection is in fact exactly what we get when running the SPI for a positive operator $\hat 1+|v\rangle\langle v|$.
	Finally, the resulting product vector serves as the initial vector for the SPI algorithm applied to $\hat L$.

	Our numerical results and comparison with other methods confirm the assumption that the constructed starting vector is sufficient, as the described procedure returns the same values.
	Still, a rigorous proof of this observation requires further investigations.
	Until then, the choice of an operator basis of starting vectors is preferable in the general case.

\subsubsection{Convergence criterion}\label{subsec:Convergence}

	The flowchart in Fig. \ref{Fig:Algorithm} requires a check for convergence in step \textcircled{\scriptsize 3}.
	Theorems \ref{Th:Local} and \ref{Th:Global} guarantee, in theory, the convergence of the SPI.
	In a practical implementation of the algorithm, however, the computer needs to know when convergence is reached in a numerical sense.

	We apply a convergence criterion that is based on the SEE.
	In the $s$th cycle of the algorithm, we obtain the vector $|\chi^{(s)}\rangle=(\hat{L}-g^{(s)}\hat 1)|a_1^{(s)},\ldots,a_N^{(s)}\rangle$; see Eq. \eqref{Eq:SEE2}.
	By definition (see Sec. \ref{Sec:Preliminaries}), $|a_1^{(s)},\ldots,a_N^{(s)}\rangle$ is a separability eigenvector if and only if $|\chi^{(s)}\rangle$ is $N$ orthogonal.
	Likewise, convergence is reached if and only if $|\chi^{(s)}\rangle$ is $N$ orthogonal.
	Theorem \ref{Th:Local} guarantees that we approach this scenario---meaning that $\lim_{s\to\infty}\langle a_1^{(s)},\ldots,a_{j-1}^{(s)},x,a_{j+1}^{(s)},\ldots,a_N^{(s)}|\chi^{(s)}\rangle=0$ for all $x\in\mathcal H_j$ and $j=1,\ldots, N$.

	In fact, this $N$-orthogonality requirement can be used to quantify the closeness to the solution.
	For this purpose, we can evaluate if the following inequality is satisfied: $\max_{j=1,\ldots,N}\max_{x\in\mathcal B_j}|\langle a_1^{(s)},\ldots,a_{j-1}^{(s)},x,a_{j+1}^{(s)},\ldots,a_N^{(s)}|\chi^{(s)}\rangle|{<}\epsilon$, for a sufficiently small $\epsilon$ and all $x\in\mathcal{B}_j$, where $\mathcal B_j$ is a basis of $\mathcal H_j$.
	The machine precision of the representation of numbers on a computer bounds the value of $\epsilon$.
	When the inequality is satisfied, the possible numerical convergence is achieved and the current iteration $|a_1^{(s)},\ldots,a_N^{(s)}\rangle$ is the desired approximation to the separability eigenvector.

\section{Benchmark}\label{Sec:Benchmark}

	We now want to find out how the SPI performs as opposed to other methods that allow the construction of arbitrary entanglement witnesses.
	A simple brute-force approach to obtain the entanglement criterion \eqref{Eq:EntanglementCondition} for $\hat{L}$ is to find all separable pure states and calculate the expectation value of $\hat{L}$.
	Then, $g_{\max}$ is the maximum of these values.

	As the search space for these vectors is over a continuum, one could use a generally applicable global optimization algorithm, such as genetic algorithms \cite{HH04}.
	This presents a state-of-the-art method to solve optimization problems.
	It is rather fast and inspired by evolutionary processes in biology.
	Thus, we implemented such a genetic algorithm to evaluate the performance of the SPI.
	A genetic algorithm requires a fitness function to be minimized, which will be $f(v)=-\langle v|\hat{L}|v\rangle$, the negative of the expectation value of $\hat{L}$.
	An intermediate step ensures that the argument vector $|v\rangle$ is indeed a product vector.
	During the runtime, the genetic algorithm will minimize $f(v)$ and converge towards a vector $|v_{0}\rangle$ with $f(v_{0})=\min_vf(v)$.
	The resulting minimization will give the maximal separability eigenvalue $g_\mathrm{max}=-f(v_0)$, or at least a close approximation.

	To show the advantages of the proposed algorithm, SPI, as opposed to this simple maximization strategy, we compare the two approaches for the following, different scenarios:
	We consider a bipartite system ($N=2$) and vary the dimensions, $d_1=d_2=d$; we fix the dimensions (here, $d_1=\cdots=d_N=2$) and increase the number of parties $N$.
	As discussed previously, we choose the convergence criterion in Sec. \ref{subsec:Convergence}.
	The starting vector is chosen as the maximal separable projection on the (standard) eigenvector corresponding to the maximal eigenvalue of $\hat{L}$.

	To exclude any bias, the chosen operators are randomly generated by first defining a random operator $\hat{M}$ acting on the $D$-dimensional space, where $D=d_1\cdot\ldots\cdot d_N$.
	Then, we construct a positive and normalized operator $\hat{L}$ for which we want to find the maximal separability eigenvalue as
	\begin{align}
		\hat{L}=\frac{1}{\mathrm{tr}(\hat{1}+\hat{M}\hat{M}^\dagger)}\left(\hat{1}+\hat{M}\hat{M}^\dagger\right).
	        \label{Eq:Random}
	\end{align}
	The SPI and brute-force approaches have been tested for 100 randomly selected operators.
	To make the runtimes comparable, the same set of random operators was used for both approaches.

\begin{figure}
	\includegraphics[width=.48\textwidth]{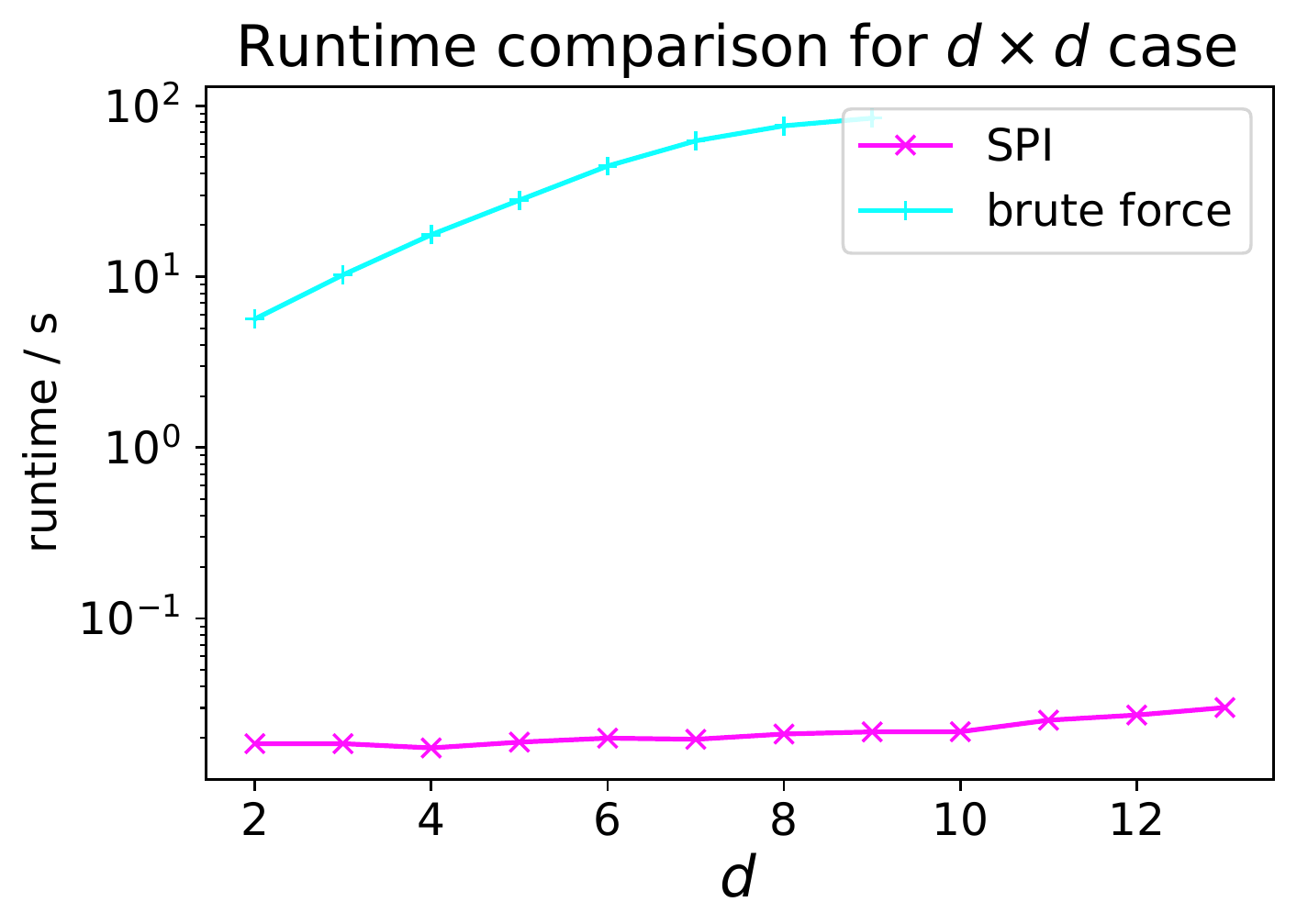}\\
	\includegraphics[width=.48\textwidth]{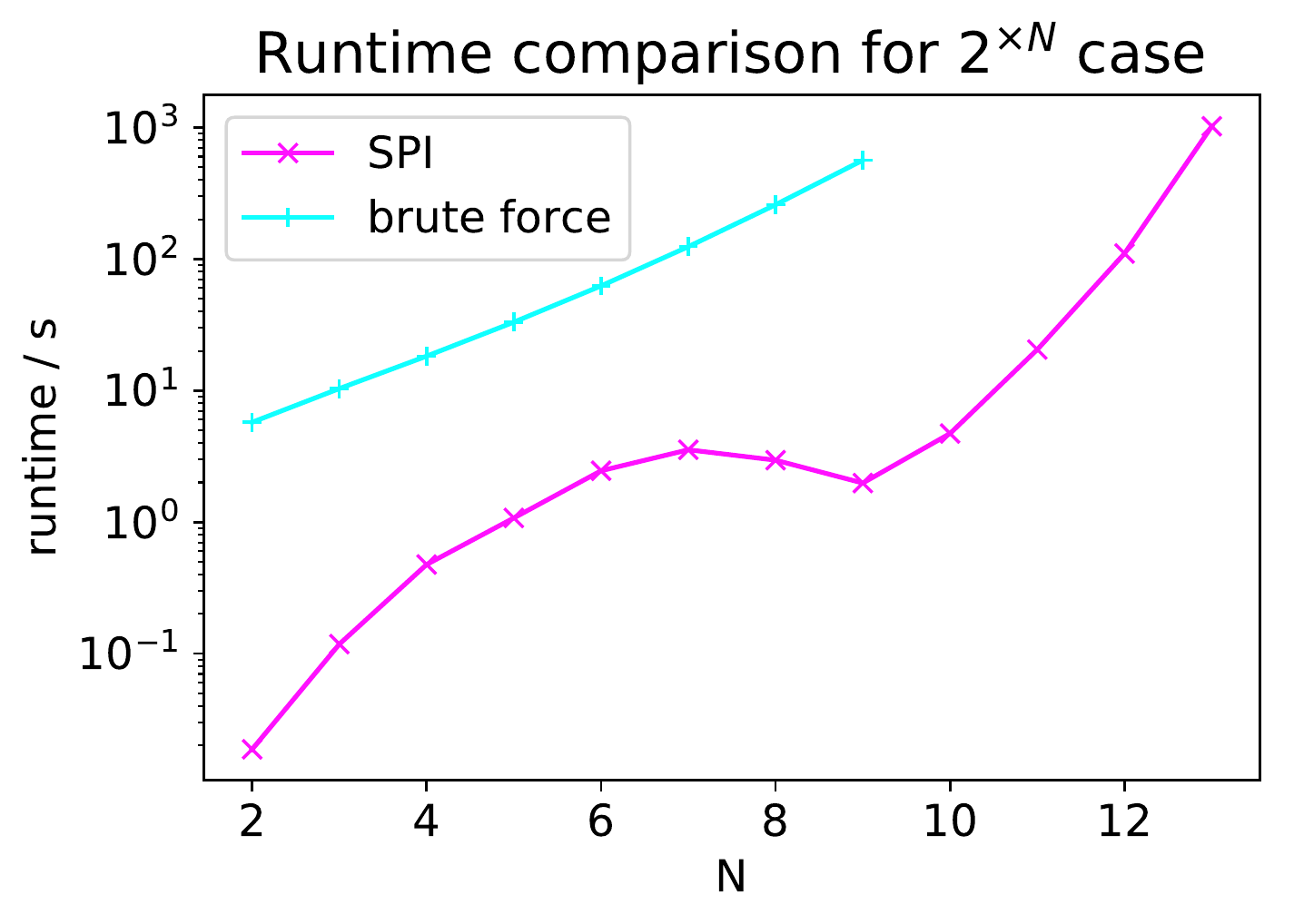}
	\caption{
		Benchmark results comparing the SPI (magenta line) to a brute-force approach, which is a genetic algorithm (cyan line).
		Top panel: The results for the runtime comparison between both approaches are shown when scaling the dimensions of each subsystem for a bipartite state, $N=2$ and $d_1=d_2=d$.
		Bottom panel: The corresponding results are shown when scaling the number $N$ of parties, which are qubit systems ($d_1=\cdots=d_N=2$).
		It can be seen that the SPI performs better than the competing approach by several orders of magnitude.
	}\label{Fig:Benchmark}
\end{figure}
	
	Figure \ref{Fig:Benchmark} shows the average runtime for the SPI compared to the brute-force approach.
	These results come from running both algorithms on a desktop computer.
	The runtime of the SPI is, on average, at least two orders of magnitude lower for the considered sample size of 100 randomly generated test operators.
	In bipartite systems (Fig. \ref{Fig:Benchmark}, top panel), we see a smaller scaling behavior of the SPI, whereas the scaling is about the same for an increasing number of qubits (bottom plot).
	Moreover, focusing on the numbers of subsystems (Fig. \ref{Fig:Benchmark}, bottom panel), we see that the SPI finds the maximal separability eigenvalue for a state acting on a 13-fold Hilbert space (dimensionality $D=2^{13}=8\,192$) in roughly the same time as the other approach manages to find in the ninefold case (dimensionality $D=2^9=512$).
	It is also worth mentioning that all curves of the presented study in Fig. \ref{Fig:Benchmark} can be roughly approximated by exponential functions of the overall dimensionality ($D=d^2$ [top] and $D=2^N$ [bottom]), representing the expected exponentially increasing runtime of the separability problem.
	The dip (in favor of the SPI) at $N=9$ in the bottom plot cannot be explained at this point and requires further investigations.

	Our benchmark indicates the superior potential of the SPI algorithm to numerically construct entanglement tests.
	Specifically, it outperforms the competing approach for high-dimensional scenarios, which includes the dimensionality of the individual parties as well as the number of parties itself.
	This enables a comparably efficient tool for the identification of entanglement in complex physical systems.
	Keep in mind that the runtimes shown in Fig. \ref{Fig:Benchmark} are from running the SPI on a desktop computer; computation clusters might improve the performance even further by a large margin.

	Beyond the genetic algorithm, there exist more specialized algorithms, treating Eq. \eqref{Eq:gMax} as a maximization of a multivariate polynomial.
	Such approaches are also NP-hard problems, meaning they can not be solved in polynomial time by a non-deterministic Turing machine, and only lower bounds of the global maximum can be found in polynomial time \cite{BPT13}.
	We apply one state-of-the-art realization of such an algorithm to find the maximum of a polynomial \cite{HLL07}, using semidefinite programming, instead of the problem of finding the maximal separability eigenvalue of an operator. Semidefinite programming is a frequently applied technique used for entanglement tests, cf., e.g., Refs.~\cite{VB96,R01,AM02,PDS02,E03,BV04,DPS05}.
	Already in a $3\times3$ case, the algorithm in Ref.~\cite{HLL07} failed to be conclusive and, in fact, returned a lower value than our SPI.
	For use as an optimal witness, the true maximal separability eigenvalue is crucial; thus, the result of the competing algorithm could lead to a false indication of entanglement.
	In all other tested cases in which the algorithm was conclusive, our SPI was superior in terms of speed and accuracy.

\section{Examples} \label{Sec:Example}

	As a proof of principle, let us apply our algorithm to detect entanglement of states of special relevance.
	Specifically, we study the two-qutrit Horodecki state \cite{HHH99} and the four-qubit Smolin state \cite{S01}.
	Both states have been classified as bound-entangled states.
	In the case of the Horodecki state, this arises from the dimensions of the state, which is acting on a $3\times3$-dimensional Hilbert space.
	The Smolin state acts on a $2\times2\times2\times2$-dimensional Hilbert space, and the bound-entangled nature arises from the fact that the state is separable with respect to all bipartitions consisting of two subsystems each, yet still entangled in all other partitions.
	By applying the SPI algorithm, we aim at confirming the weak entanglement properties of those bound entangled states for which the well-known partial transposition test \cite{P96,H97} fails to be conclusive.

	The first example, the Horodecki state, is defined as \cite{HHH99}
	\begin{align}
		\hat{\rho}_\alpha=\frac{1}{7}\left(2|\Psi\rangle\langle\Psi|+\alpha\hat{\sigma}_++(5-\alpha)\hat\sigma_-\right),
	        \label{Eq:2SpinState}
	\end{align}
	where $\hat{\sigma}_+=(|0,1\rangle\langle0,1|+|1,2\rangle\langle1,2|+|2,0\rangle\langle2,0|)/3$ and $\hat{\sigma}_-=(|1,0\rangle\langle1,0|+|2,1\rangle\langle2,1|+|0,2\rangle\langle0,2|)/3$ are separable and $|\Psi\rangle=(|0,0\rangle+|1,1\rangle+|2,2\rangle)/\sqrt 3$ is the entangled contribution.
	The parameter can be chosen as $0\leq\alpha\leq5$; otherwise the density operator $\hat{\rho}_\alpha$ does not represent a physical state.
	The Horodecki state was shown to be entangled for $\alpha>3$ and $\alpha<2$ \cite{HHH99}.

	For our entanglement analysis based on the criterion \eqref{Eq:EntanglementCondition}, a positive-definite, Hermitian operator $\hat{L}$ is required.
	For simplicity, the test operator will be chosen as $\hat{L}_\beta=\hat{\rho}_\beta$.
	We calculate the maximal separability eigenvalues $g_\beta$ for every $\hat{L}_\beta$ where $0\leq\beta\leq5$.
	In this entanglement test, a state is verified to be entangled if
	\begin{align}
	        g_\beta-\mathrm{tr}(\hat{\rho}_\alpha\hat{L}_\beta)<0,
	        \label{Eq:Exampleentangled}
	\end{align}
	which corresponds to the criterion based on the entanglement witness $\hat W_\beta=g_\beta\hat 1-\hat L_\beta$.

\begin{figure}
	\includegraphics[width=0.5\columnwidth]{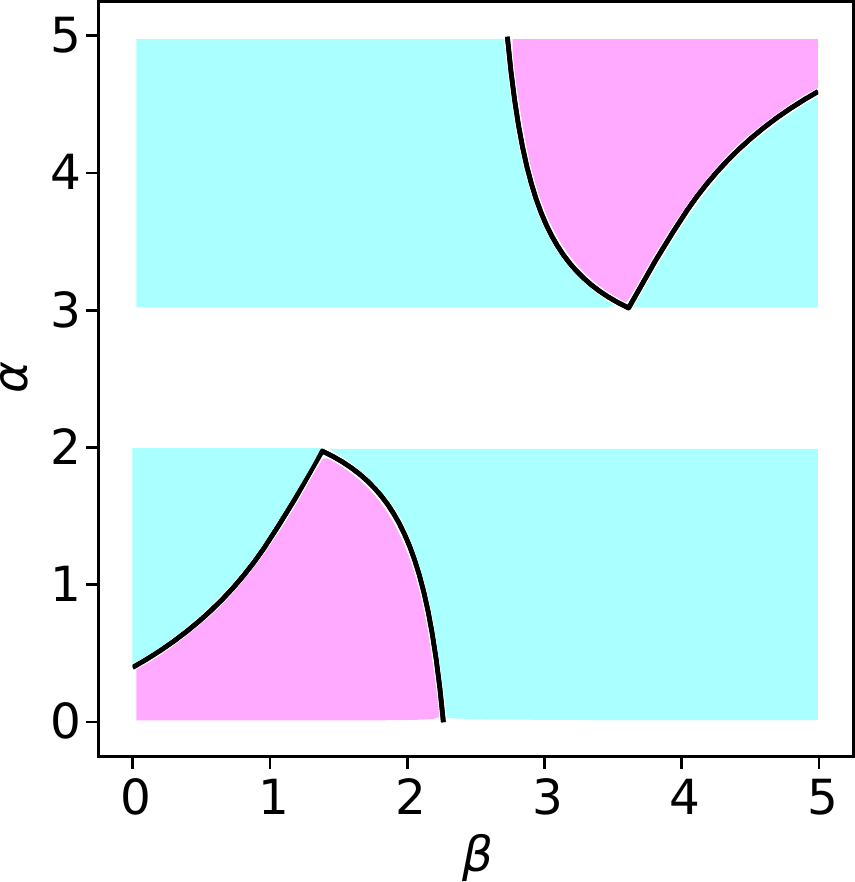}
	\caption{
		Results of the entanglement test for the bound-entangled, two-qutrit Horodecki state.
		No state $\hat{\rho}_{\alpha}$ with $2\leq\alpha\leq3$ (blank area) has been detected as entangled.
		In the magenta areas, the criterion Eq. \eqref{Eq:Exampleentangled} certifies entanglement for the given combination of $\alpha$ and $\beta$, which does not hold true for the cyan areas.
	}\label{Fig:Results}
\end{figure}

	The results are shown in Fig. \ref{Fig:Results}.
	The entanglement criterion \eqref{Eq:Exampleentangled} is satisfied in the magenta colored areas.
	The blank area corresponds to parameters $2\leq \alpha\leq 3$ for which no entanglement could be detected, which agrees with the prediction in Ref. \cite{HHH99}.
	In all other cases (cyan area), there exists at least one other value $\beta'$ for which $\hat{L}_{\beta'}$ verifies entanglement.
	Thus, we correctly and straightforwardly certify entanglement of all Horodecki states, which are positive under partial transposition, using our SPI approach.

	Beyond the bipartite case, let us apply our method to the multipartite scenario for which the partial transposition criterion does not apply in principle.
	For this reason, we study the four-partite Smolin state \cite{S01},
	\begin{align}\label{Eq:originalSmolin}
		\hat{S} = \frac{1}{16}\left(\hat{1}+\hat{\sigma}_x^{\otimes 4}+\hat{\sigma}_y^{\otimes 4}+\hat{\sigma}_z^{\otimes 4}\right),
	\end{align}
	where $\hat\sigma_x$, $\hat\sigma_y$, and $\hat\sigma_z$ denote the Pauli spin matrices.
	We restrict ourselves to a test operator of the simple form $\hat L=\hat S$.

	In the multipartite case, we can analyze different forms of entanglement, such as bipartitions, tripartitions, and four-partitions for the state under study.
	In total, we have 14 partitions.
	However, because of the symmetry, cf. Eq. \eqref{Eq:originalSmolin}, we can restrict ourselves to the bipartitions $\{1\}{:}\{2,3,4\}$ and $\{1,2\}{:}\{3,4\}$, the tripartition $\{1\}{:}\{2\}{:}\{3,4\}$, and the four-partition $\{1\}{:}\{2\}{:}\{3\}{:}\{4\}$.

	The SPI algorithm was run for all partitions.
	The results are listed in Table \ref{Tab:Smolin}.
	For applying entanglement criterion Eq. \eqref{Eq:EntanglementCondition}, we additionally compute $\mathrm{tr}(\hat{L}\hat{S})=1/4$.
	Thus, in agreement with the results in Ref. \cite{S01}, entanglement could be verified for all partitions, except for the bipartition which consists of two subsystems each, i.e., $\{1,2\}{:}\{3,4\}$.
	
\begin{table}
	\caption{
		Separability eigenvalues of the operator $\hat{L}=\hat S$ [Eq. \eqref{Eq:originalSmolin}].
		The maximal separability eigenvalues $g_\mathrm{max}$ are listed for the corresponding partitions.
	}\label{Tab:Smolin}
	\begin{tabular}{c c}
		\hline\hline
		\quad \quad \quad Partition \quad \quad \quad  & \quad \quad $g_\mathrm{max}$ \quad \quad
		\\\hline
		$\{1,2\}{:}\{3,4\}$ & 0.250
		\\
		$\{1\}{:}\{2,3,4\}$ & 0.125
		\\
		$\{1\}{:}\{2\}{:}\{3,4\}$ & 0.125
		\\
		$\{1\}{:}\{2\}{:}\{3\}{:}\{4\}$ & 0.125
		\\\hline\hline
\end{tabular}
	\end{table}

	In this section, we demonstrated the direct application of our SPI algorithm to construct entanglement probes, for example, to identify bound instances of entanglement.
	We deliberately chose such weakly entangled states, which have been characterized previously to challenge our method and compare our numerical results with sophisticated exact analysis.
	In particular, entanglement was verified in bipartite qudit and multipartite qubit states.
	The entanglement of the states under study is a challenge for other directly applicable methods as the partial transposition criterion gives inconclusive results.

\section{Discussion}\label{Sec:discussion}

	We introduced, implemented, and applied a method to numerically construct entanglement tests.
	In this section, let us discuss how this technique can be used in experiments, how it improves other entanglement probes, and how it can be generalized to detect other forms of entanglement.
	Finally, we discuss future research directions that become accessible with our approach and address the interdisciplinary importance of the introduced technique by relating it to a current problem in pure mathematics.

\subsection{Experimental implementation}

	A major benefit of our approach is the direct applicability in experiments.
	Suppose that the set of observables $\{\hat M_k:k=1,\ldots,m\}$ describes a measurement scheme.
	In other words, the data yield the expectation values $\langle\hat M_k\rangle=\mathrm{tr}(\hat M_k\hat\rho)$.
	An example for such operators relates to a displaced photon-number correlation \cite{KVS17}.
	In general, a family of positive operators $\hat L$ can be constructed from the considered measurements,
	\begin{align}
		\hat L=\nu\hat 1+\sum_{k=1}^m\mu_k\hat M_k,
	\end{align}
	by choosing real-valued coefficients $\mu_k$ and adjusting $\nu$ to ensure positivity of $\hat L$.

	The entanglement criterion \eqref{Eq:EntanglementCondition} can be applied.
	On the one hand, the experimental expectation value is given by $\langle \hat L\rangle=\nu+\sum_{k=1}^m \mu_k\langle \hat M_k\rangle$.
	On the other hand, we get the maximal expectation value for separable states, $g_{\max}$, from the application of our SPI to the family of operators $\hat L$ under study.
	Note that a variation over the coefficients $\mu_k$ also enables an optimal entanglement verification based on the set of measured observables, similarly to the technique applied to Gaussian measurements in Refs. \cite{GSVCRTF15,GSVCRTF16}.

\subsection{Relations to other entanglement criteria}

	As mentioned earlier, our entanglement criteria are identical to witnesses [Eq. \eqref{Eq:witnessconstruction}].
	Furthermore, based on the Choi-Jamio\l{}kowski isomorphism \cite{C75,J72}, entanglement witnesses enable the formulation of positive, but not completely positive maps to probe entanglement \cite{HHH96,HHH01}.
	Thus, our numerical method can be used to construct previously unknown families of such maps.
	For instance, the test operators that verified the entanglement of the bound-entangled states (Sec. \ref{Sec:Example}) necessarily lead to maps that go beyond the partial transposition since the partial transposition cannot detect the entanglement of states considered in those examples.

	In addition, in Ref. \cite{RH14}, an elegant approach was formulated that enables the construction of device-independent entanglement witnesses from device-dependent ones.
	This technique is based on a matrix-product extension that assigns to each subsystem an auxiliary Hilbert space, but requires the previous knowledge of a witness.
	Such desired initial witnesses can be provided by our algorithm and combined with the method from Ref. \cite{RH14} to construct device-independent entanglement witnesses.

\subsection{Outlook}

	Beyond the witnessing of multipartite entanglement, the SEE approach has been generalized.
	Thus, let us briefly discuss some future generalizations of our numerical method for the aim of exploring entanglement in a broader context.
	
	The detection of $K$-entanglement, and thus of \textit{genuine entanglement}, is possible by finding the maximum of all maximal separability eigenvalues for an operator with respect to partitions of the length $K$ \cite{GSVCRTF16}.
	It is therefore a straightforward extension to the SPI to find the optimal witness for $K$-entanglement with the introduced algorithm---the algorithm is run multiple times for different partitions and the maximum of the results is the required separability eigenvalue.

	Furthermore, some physical problems require solutions of a generalized EE, $\hat L|\Phi\rangle=\lambda\hat P|\Phi\rangle$, where the right-hand side includes a contribution that is different from the identity, $\hat P\neq\hat 1$.
	Interestingly, the same holds true for the SEE.

	One example is the verification of entanglement in systems of indistinguishable particles, which is based on a generalized SEE and where $\hat P$ represents the (anti)symmetrization operator for bosons (fermions) \cite{RSV15}.
	Another example is the quantification of multipartite entanglement via generalized Schmidt-number witnesses \cite{SSV14}.
	There, $\hat P$ takes the form of a spinor projection (details can be found in the supplement to Ref. \cite{SSV14}).
	A third example is the detection of multipartite entanglement in systems for which the number of subsystems is not fixed.
	For instance, the underlying generalized SEE applies to the construction of multiparticle-entanglement witnesses for fluctuating particle numbers \cite{SW17}.

	Thus, a generalization of the SPI to account for such generalized SEEs, including $\hat P$, will further enhance the range of applications.
	It is worth mentioning that the desired generalization is well known for the PI, which is likely to be applicable to the SPI in a similar manner.

	Furthermore, the standard EE applied to the density operator leads to the spectral decomposition of the state.
	Similarly, the SEE can be used to expand the density operator in terms of separability eigenvectors and a quasiprobability distribution \cite{SV09quasi,SW18}.
	The latter one includes negativities iff the state is entangled; see Ref. \cite{TBV17} for an application to uncover bound entanglement.
	However, this approach requires the computation of all separability eigenvectors.
	Therefore, similar to the subspace iteration for the PI, a generalization of the SPI to include all solutions, beyond the one that corresponds to the maximal separability eigenvalue, could lead to a broader applicability of entanglement quasiprobabilities.

	As a final example let us consider the dynamics of quantum systems, which is described by the Schr\"odinger equation.
	To distinguish the entanglement-generating evolution from the separable dynamics, we recently introduced the separability Schr\"odinger equations \cite{SW17dyn}, which relate to the SEE in the static case.
	Again, the SPI can be the starting point for the numerical implementation of this approach.

	Thus, generalizations of the SPI have the potential to uncover multipartite entanglement in a much broader sense.
	Beyond the already-available construction of positive, but not completely positive maps and device-independent entanglement witnesses, our numerical approach builds the foundation for the future studies of entanglement.
	
\subsection{Relations to mathematical problems}

	The question of positive polynomials is an interesting and, in the most general case, unsolved mathematical problem, which has been studied for a long time \cite{B79} and finds many applications \cite{L09}.
	As already indicated in Sec. \ref{Sec:Benchmark}, any entanglement witness can be characterized by the non-negativity of a multivariate polynomial \cite{BPT13}.
	All entanglement witnesses can be generated through the solution of the SEE.
	Therefore, the solution of the SEE enables the construction and characterization of positive multivariate polynomials; see also Appendix \ref{App:SEE}.
	Consequently, the proposed SPI is an alternative approach to numerically solving the positivity problem of polynomials.

	Another family of important problems in pure mathematics that could benefit from the SPI are partial differential equations, which are also closely related to many problems in physics.
	For instance, the applicability of the method of separation of variables corresponds to the question of whether or not solutions are factorizable, i.e., a tensor product.
	Since a separable eigenfunction is also a separability eigenfunction \cite{SV09}, i.e., eigenvector in the function space, the SPI can be applied to find factorizable solutions of the partial differential equation.

	Moreover, nonlinear partial differential equations address questions such as finding the ground state to a nonlinear energy functional.
	If this functional is polynomial, a problem related to the previously mentioned characterization of multivariate polynomials can be formulated.
	Namely, the numerical approximation to the ground state can be obtained by the multipartite SPI as the maximum of the negative nonlinear energy functional, resulting in the minimal energy.

\section{Conclusion}\label{Sec:Conclusion}

	In this paper, we introduce an algorithm, the SPI, to numerically construct arbitrary multipartite entanglement witnesses.
	This algorithm enables us to find the maximal separability eigenvalues, which directly results in measurable entanglement tests.
	Beyond the formulation of our method, we also provide the mathematical background for the SPI, which yields the maximal solution of the nonlinear separability eigenvalue problem addressing the complex entanglement problem in quantum physics.
	Furthermore, our framework is supplemented by performing a benchmark of our approach, applying it to uncover hard-to-detect forms of entanglement, and relating it to other methods in the theory of quantum entanglement and their experimental application.

	Our algorithm shows two crucial steps---namely, forward and backward iteration---following directly from the cascaded structure of the separability eigenvalue equations.
	The forward iteration reduces the number of parties until we have a single-party problem, which is then used in the backward iteration to solve the multipartite problem.
	This property also allows us to prove the convergence of the SPI to reliably produce entanglement tests based on arbitrary observables.
	Interestingly, our algorithm includes the well-known power iteration, which is able to calculate the maximal (standard) eigenvalue, as a special case.

	We show the efficiency of our approach in comparison with another method, which is mainly based on a genetic algorithm.
	The genetic algorithms presents a state-of-the-art approach to solve arbitrary optimization problems.
	The SPI is faster by two orders of magnitude, which is partly because of its directed design to specifically address the entanglement problem.
	For example, we analyze the runtime as a function of the dimension of a bipartite quantum system.
	In addition, we numerically solve the separability eigenvalue equations in a feasible time for operators up to a 13-party qubit Hilbert space, corresponding to 8\,192 dimensions.

	Furthermore, we apply the SPI to bound-entangled states whose entanglement detection is a cumbersome problem.
	For instance, the frequently applied partial transposition criterion fails to uncover the entanglement of the considered examples.
	Applying the SPI, we straightforwardly verify this weak form of entanglement, proving the advantage of our method.
	Moreover, we demonstrate with these examples that our algorithm renders it possible to uncover entanglement of all forms of partial entanglement in multipartite systems.
	It is also worth mentioning that entanglement of continuous-variable systems can be detected in finite subspaces, allowing us to apply our algorithm to these kind of states as well.

	We outline the versatile nature of our method and its impact on future research by relating it to other open problems in quantum entanglement and beyond.
	For instance, the construction of entanglement witnesses, which is achieved by our SPI, is the basis for the formulation of positive, but not completely positive maps for entanglement detection and the construction of device-independent entanglement witnesses.
	Furthermore, we describe the construction of entanglement criteria based on measured quantities and outline several generalizations, which are---at their core---related to our method.

	Thus, we devise a relatively simple, yet versatile approach to numerically construct entanglement tests in multipartite systems.
	The direct implementation of our method enables us to certify complex forms of quantum correlations based on measurable criteria.
	In addition, we derive the required mathematical background of our algorithm to ensure its operation and benchmark its performance.
	To the best of our knowledge, there exists no alternative method of entanglement verification that is applicable to complex systems that our method can manage.
	To summarize, we provide a full numerical framework for the detection of multipartite entanglement for theoretical studies and, more importantly, for application in current and future experiments using entanglement in quantum information and communication protocols.

\begin{acknowledgments}
	This work has received funding from the European Union's Horizon 2020 research and innovation program under Grant Agreement No. 665148 (QCUMbER).
\end{acknowledgments}

\appendix
\setcounter{theorem}{0}
\section{Proof of Theorem \ref{Th:Forward}}\label{App:Forward}
	\begin{theorem}[Forward iteration]
		Let $|a_1,\ldots,a_N\rangle$ be the separability eigenvector corresponding to the maximal separability eigenvalue of a positive $N$-partite operator $\hat{L}$.
		Furthermore, let $|\Psi\rangle=\hat{L}|a_1,\ldots,a_N\rangle$.
		For the $(N-1)$-partite operator $\hat{L}'=\mathrm{tr}_N(|\Psi\rangle\langle\Psi|)$, the equality
		\begin{align}
			\label{Eq:THM1Proof}
		\begin{aligned}
			&\langle a_1,\ldots,a_N|\hat{L}|a_1,\ldots,a_N\rangle
			\\
			=&\sqrt{\langle a_1,\ldots,a_{N-1} |\hat{L}'|a_1,\ldots,a_{N-1}\rangle}
		\end{aligned}
		\end{align}
		holds true.
	\end{theorem}

	\begin{proof}
		As a shorthand notation, let $|v_N\rangle=|a_1,\ldots,a_N\rangle$.
		Using the cascaded structure (CS) and the abbreviation $\hat{L}'=\mathrm{tr}_N|\Psi\rangle\langle\Psi|$, the statement is derived as follows:
		\begin{align*}
		     \underset{|v_N\rangle\in\mathcal S}{\mathrm{max}}\langle v_N|\hat{L}|v_N\rangle
			&=	\underset{|v_N\rangle,|v'_N\rangle\in\mathcal S}{\mathrm{max}}\langle v'_N|\underbrace{\hat{L}|v_N\rangle}_{=:|\Psi(v_N)\rangle}\\
			&=	\underset{|v_N\rangle\in\mathcal S}{\mathrm{max}}\,\underset{|v'_N\rangle\in\mathcal S}{\mathrm{max}}\langle v'_N|\Psi(v_N)\rangle\\
			&=	\underset{|v_N\rangle\in\mathcal S}{\mathrm{max}}\,\sqrt{\underset{|v'_N\rangle\in\mathcal S}{\mathrm{max}}\left(\langle v'_N|\Psi(v_N)\rangle\right)^2}\\
			&=	\underset{|v_N\rangle\in\mathcal S}{\mathrm{max}}\,\sqrt{\underset{|v'_N\rangle\in\mathcal S}{\mathrm{max}}\langle v'_N|\Psi(v_N)\rangle\langle\Psi(v_N)|v'_N\rangle}\\
			&\overset{\mathrm{CS}}{=}	\underset{|v_N\rangle\in\mathcal S}{\mathrm{max}}\,\sqrt{\underset{|v'_{N-1}\rangle\in\mathcal S}{\mathrm{max}}\langle v'_{N-1}|\hat{L}'(v)|v'_{N-1}\rangle},
		\end{align*}
		where we chose global phases such that scalar products correspond to non-negative numbers.
	\end{proof}
	
\section{Proof of Theorem \ref{Th:Monotony}}\label{App:Monotony}
\addtocounter{theorem}{1}
	\begin{theorem}[Monotony]
		Let $|a_1,\ldots,a_N\rangle\in\mathcal S$ and $|a'_1,\ldots,a'_N\rangle\in\mathcal{S}$ such that
	        \begin{align}
	                |a'_1,\ldots,a'_N\rangle=\underset{|b_1,\ldots,b_N\rangle\in\mathcal S}{\mathrm{arg\,max}}\langle b_1,\ldots,b_N|\hat{L}|a_1,\ldots,a_N\rangle.
	        \end{align}
		Then the inequality
		\begin{align}
			\label{Eq:MonotonyProof}
		\begin{aligned}
			&\langle a_1,\ldots,a_N|\hat{L}|a_1,\ldots,a_N\rangle
			\\
			\leq&
			\langle a'_1,\ldots,a'_N|\hat{L}|a_1,\ldots,a_N\rangle
			\\
			\leq&
			\langle a'_1,\ldots,a'_N|\hat{L}|a'_1,\ldots,a'_N\rangle
		\end{aligned}
		\end{align}
		holds true.
		Furthermore, the equality in Eq. \eqref{Eq:MonotonyProof} holds true iff $|a'_1,\ldots,a'_N\rangle=| a_1,\ldots,a_N\rangle$.
	\end{theorem}

	\begin{proof}
		The inequality
		\begin{align*}
			\langle a_1,\ldots,a_N|\hat{L}|a_1,\ldots,a_N\rangle
			\leq
			\langle a'_1,\ldots,a'_N|\hat{L}|a_1,\ldots,a_N\rangle
		\end{align*}
		directly follows from the definition of $|a'_1,\ldots,a'_N\rangle$.
		The second inequality
		\begin{align*}
			\langle a'_1,\ldots,a'_N|\hat{L}|a_1,\ldots,a_N\rangle
			\leq
			\langle a_1',\ldots,a_N'|\hat{L}|a_1',\ldots,a_N'\rangle
		\end{align*}
		can be proved using the Cauchy-Schwarz inequality (CSI). 
		As a shorthand, let us define $|v_N\rangle=|a_1,\ldots,a_N\rangle$ and $|v'_N\rangle=|a'_1,\ldots,a'_N\rangle$ and consider the $\hat{L}$-induced scalar product $\langle v|\rangle v_{\hat{L}}=\langle v|\hat{L}|v\rangle$:
		\begin{align*}
			\langle v_N|v_N\rangle_{\hat{L}}^2&\leq
			\langle v'_N|v_N\rangle_{\hat{L}}^2\overset{\mathrm{CSI}}{\leq}
				\langle v_N|v_N\rangle_{\hat{L}}\langle v'_N|v'_N\rangle_{\hat{L}}\\
			\Leftrightarrow
			\langle v_N|v_N\rangle_{\hat{L}}&\leq
				\langle v'_N|v'_N\rangle_{\hat{L}}\\
			\Rightarrow
			\langle v'_N|v_N\rangle_{\hat{L}}^2&\leq
				\langle v'_N|v'_N\rangle_{\hat{L}}\langle v'_N|v'_N\rangle_{\hat{L}}\\
			\Rightarrow
			\langle v'_N|v_N\rangle_{\hat{L}}&\leq
				\langle v'_N|v'_N\rangle_{\hat{L}}\\
		\end{align*}
		Here, the second row follows from reduction by $\langle v_N|v_N\rangle_{\hat{L}}$; the third row can be found by substituting the inequality in row two into the right side of the inequality in row one.
		Note that the equality holds if and only if $|v_N\rangle\parallel|v'_N\rangle$.
	\end{proof}

\section{Proof of Theorem \ref{Th:Local}}\label{App:Local}

	\begin{theorem}[Local convergence]
	        For any starting vector, the sequence $(g^{(s)})_s$ of expectation values generated by the SPI converges, i.e., the limit
	        \begin{align}
			\label{Eq:LocalConvergenceProof}
	                \lim_{s\to\infty} g^{(s)} = \bar{g}
	        \end{align}
		exists and is bounded as $0\leq\bar{g}\leq g_{\mathrm{max}}$.
	\end{theorem}

	\begin{proof}
		The state $|v_N\rangle:=|a_1^{(s)},\ldots,a_N^{(s)}\rangle$ is separable for any $s$, where $s$ indexes the iteration steps of the SPI.
		Further, let $|v'_N\rangle:=|a_1^{(s+1)},\ldots,a_N^{(s+1)}\rangle$ be the next approximation to the separability eigenvector corresponding to an optimal separability eigenvalue. 
		By design, $\langle v'_N|\hat{L}|v_N\rangle\rightarrow\mathrm{max}$ holds such that Theorem \ref{Th:Monotony} applies.
		Thus, the sequence $(g^{(s)})_s$ is monotonous.
		Furthermore, as $\hat{L}$ is a bounded operator, the sequence is also bounded.
		By definition of a convergent series, $(g^{(s)})_s$ converges to, at least, a local maximum.
	\end{proof}
	
\section{Proof of Theorem \ref{Th:Global}}\label{App:Global}

	\begin{theorem}[Global convergence]
		Let $\Sigma$ be a set of separable starting vectors and $(g_{\Phi}^{(s)})_s$ be sequences of expectation values generated by the SPI for a starting vector $|{\Phi}\rangle\in\Sigma$.
		Further, say $\{\bar{g}_{\Phi}\in\mathbb{R}:|{\Phi}\rangle\in\Sigma \text{ and } \bar{g}_{\Phi}=\lim_{s\to\infty} g_{\Phi}^{(s)}\}$ defines the set of optimal expectation values (limits of the converged sequences) for each starting vector.
		The maximal separability eigenvalue for the operator $\hat{L}$ is $g_\mathrm{max}=\mathrm{max}_{\Phi\in\Sigma}\{\bar{g}_{\Phi}\}$, which is the maximum of the limit to the series of expectation values for each starting vector.
	\end{theorem}

	\begin{proof}
		The global convergence of the SPI is shown via proof by induction over the number of subsystems $N$.
		The expression $\hat{L}^{(i)}$ denotes an operator acting on a composition of $i$ Hilbert spaces.
		Further, we use $|v_i\rangle=|a_1,\ldots,a_i\rangle$ and $\bar{g}_i=\lim_{s_1,\ldots,s_i\to\infty}\langle a_1^{(s_1)},\ldots,a_i^{(s_i)}|\hat{L}^{(i)}|a_1^{(s_1)},\ldots,a_i^{(s_i)}\rangle$ as the optimal expectation value of the $i$th subsystem over separable states, with $s_i$ counting the iterations of the SPI in the $i$th subsystem.

	\paragraph*{Basis of induction.}---
		For $N=1$, the SPI is the PI for which the convergence is well known \cite{GvL13}.
		The optimal expectation value for the one-subsystem operator $\hat{L}^{(1)}$ can be found as
		\begin{align}\label{Eq:Induction_Basis}
			\bar{g}_1=\lim_{s_1\to\infty}\langle a_1^{(s_1)}|\hat{L}^{(1)}|a_1^{(s_1)}\rangle,
		\end{align}
		where $|a_1^{(s_1)}\rangle=\hat{L}^{(1)}|a_1^{(s_1-1)}\rangle/\|\hat{L}^{(1)}|a_1^{(s_1-1)}\rangle\|$ and $\||\psi\rangle\|=\langle\psi|\psi\rangle^{1/2}$.

	\paragraph*{Induction hypothesis.}---
		The induction hypothesis reads
		\begin{align}\label{Eq:Induction_Hypothesis}
			\bar{g}_N=\lim_{s_N\to\infty}\langle a_N^{(s_N)}|\hat{L}_{a_1,\ldots,a_{N-1}}^{(N)}|a_N^{(s_N)}\rangle,
		\end{align}
		where $|a_N^{(s_N)}\rangle=\hat{L}_{a_1,\ldots,a_{N-1}}^{(N)}|a_N^{(s_N)}\rangle/\|\hat{L}_{a_1,\ldots,a_{N-1}}^{(N)}|a_{N}^{(s_N)}\rangle\|$.

	\paragraph*{Induction step.}---
		Under the assumption of convergence in $N-1$ subsystems [replacing $N$ by $N-1$ in the induction hypothesis, Eq. \eqref{Eq:Induction_Hypothesis}], we show convergence of the SPI in the $N$th subsystem,
		\begin{align}\label{Eq:Induction_Proof}
		        \bar{g}_N&=
				\underset{a_1,\ldots,a_{N}}{\mathrm{max}}\langle a_1,\ldots,a_N|\hat{L}^{(N)}|a_1,\ldots,a_N\rangle\notag\\
			&=
				\underset{a_1,\ldots,a_{N-1}}{\mathrm{max}}\underset{a_N}{\mathrm{max}}\,\langle a_1,\ldots,a_N|\hat{L}^{(N)}|a_1,\ldots,a_N\rangle
		\end{align}
		In the SPI algorithm, we then define
		\begin{align}\label{Eq:Induction_Forward}
		        |\Psi\rangle=\hat{L}^{(N)}|a_1,\ldots,a_N\rangle
		\end{align}
		and calculate
		\begin{align}\label{Eq:Induction_Backward}
		        |b_N\rangle = \frac{\langle a_1,\ldots,a_{N-1},\,\cdot\,|\Psi\rangle}{||\langle a_1,\ldots,a_{N-1},\,\cdot\,|\Psi\rangle||}.
		\end{align}
		Using these definitions in the calculation of $\bar{g}_N$, we get
		\begin{align*}
			\bar{g}_N=&
				\underset{a_1,\ldots,a_{N-1}}{\mathrm{max}}\underset{a_N}{\mathrm{max}}\,\langle a_1,\ldots,a_N|\hat{L}^{(N)}|a_1,\ldots,a_N\rangle\\
			=&
				\underset{a_1,\ldots,a_{N-1}}{\mathrm{max}}\,\underset{\gamma_N,a_N}{\mathrm{max}}\,\langle \gamma_N|\hat{L}_{a_1,\ldots,a_{N-1}}^{(N)}|a_N\rangle\\
			=&\sqrt{
				\underset{a_1,\ldots,a_{N-1}}{\mathrm{max}}\,
				\underset{\gamma_N,a_N}{\mathrm{max}}\,
				|\langle a_1,\ldots,a_{N-1},\gamma_N|\Psi\rangle|^2
			},
		\end{align*}
		where the second line follows from Theorem \ref{Th:Monotony}.
		By construction---following the induction step---convergence has been reached for the subsystems up to and including ${N-1}$, which leaves a maximization for $|a_N\rangle$ and $|\gamma_N\rangle$,
		\begin{align*}
			\bar{g}_N
			&=\sqrt{\underset{\gamma_N,a_N}{\mathrm{max}}\,\langle a_1,\ldots,a_{N-1},\gamma_N|\Psi\rangle\langle\Psi|a_1,\ldots,a_{N-1},\gamma_N\rangle}.
		\end{align*}
		The solution to this maximization problem is found via the cascaded structure and is equal to $|b_N\rangle$ [see Eq. \eqref{Eq:Induction_Backward}],
		\begin{align*}
			\bar{g}_N
			&=\langle a_1,\ldots,a_{N-1},b_N|\hat{L}^{(N)}|a_1,\ldots,a_{N-1},b_N\rangle.
		\end{align*}
		
		We use the induction hypothesis, Eq. \eqref{Eq:Induction_Hypothesis}, to solve the problem of finding the states $|a_1\rangle,\ldots,|a_{N-1}\rangle$. 
		Then we need to maximize $\hat{L}_{a_1,\ldots,a_{N-1}}^{(N)}$.
		Since this is an operator in one subsystem, the PI can be applied to maximize the expectation value.
		This is shown in the induction hypothesis.
		As the PI is guaranteed to converge, Eq. \eqref{Eq:Induction_Hypothesis} will indeed return a separable vector which optimizes the expectation value of $\hat{L}^{(N)}$.
		Thus, for a single starting vector, the SPI finds a separability eigenvector, which might correspond to the maximal separability eigenvalue.

		Convergence towards the separability eigenvector corresponding to the globally maximal separability eigenvalue is guaranteed by the choice of starting vectors.
		The operator basis is chosen as a set of starting vectors after every forward iteration.
		The PI converges towards the dominant eigenvalue of a matrix for a given starting vector, if the decomposition of the starting vector into the eigenbasis of the matrix has a nonzero contribution of the eigenvector corresponding to the maximal eigenvalue.
		As the operator basis spans the considered operator space, the separability eigenvector will have a nonzero contribution to the decomposition of at least one of the starting vectors.
	\end{proof}

\section{Brief derivation of the SEEs}\label{App:SEE}

	For a self-consistent reading of the present contribution, we review the derivation of the multipartite separability eigenvalue equations (see Ref. \cite{SV13}).
	Here, the derivation is based on an equivalent approach (see Ref. \cite{RSV15}), which relies on the Rayleigh quotient and is also the main idea behind the PI.

	The (multipartite) Rayleigh quotient reads
	\begin{align}
		\label{Eq:Rayleigh}
		R_{\hat{L}}(a_1,\ldots,a_N):=\frac{\langle a_1,\ldots,a_N|\hat{L}|a_1,\ldots,a_N\rangle}{\langle a_1,\ldots,a_N|a_1,\ldots,a_N\rangle},
	\end{align}
	which is the expectation value of operator $\hat{L}$ for a possibly unnormalized vector $|a_1,\ldots,a_N\rangle$.
	To relate $R$ to multivariate polynomials, we can think of $|a_j\rangle$ in terms of wave functions being Taylor-expanded in terms of polynomials of the order $d_j-1$.
	Thus, we can conclude that the desired task of maximizing the Rayleigh quotient is equal to both, maximizing a multivariate polynomial and finding the maximal expectation value of $\hat{L}$ with respect to separable states, i.e., finding its maximal separability eigenvalue.
	
	The optimal values of the Rayleigh quotient in Eq. \eqref{Eq:Rayleigh} are found for
	\begin{align}
	\begin{aligned}
		\label{Eq:PartDerrivatives}
		0=&\frac{\partial R_{\hat L}(a_1,\ldots,a_N)}{\partial\langle a_j|}
		\\
		=&\frac{\hat L_{a_1,\ldots,a_{j-1},a_{j+1},\ldots,a_N}|a_j\rangle}{\langle a_1,\ldots,a_N|a_1,\ldots,a_N\rangle}
		-g\frac{|a_j\rangle}{\langle a_j|a_j\rangle}
	\end{aligned}
	\end{align}
	for $j=1,\ldots, N$, where we use the notation $g=R_{\hat L}(a_1,\ldots,a_N)$ and the so-called reduced operator $\hat L_{a_1,\ldots,a_{j-1},a_{j+1},\ldots,a_N}=\langle a_1,\ldots,a_{j-1},\,\cdot\,,a_{j+1},\ldots,a_N|\hat L|a_1,\ldots,a_{j-1},\,\cdot\,,a_{j+1},\ldots,a_N\rangle$, acting solely on the $j$th subsystem (cf. Refs. \cite{SV13,RSV15}).

	As the Rayleigh quotient is invariant under the norm of the vector, we may assume $\langle a_j|a_j\rangle=1$.
	Consequently, the optimization of the Rayleigh quotient [cf. Eq. \eqref{Eq:PartDerrivatives}] yields the SEE in the first form as
	\begin{align}
		\label{Eq:DerivedSEE}
		\hat{L}_{a_1,\ldots,a_{j-1},a_{j+1},\ldots,a_N}|a_j\rangle &= g|a_j\rangle
	\end{align}
	for $j=1,\ldots,N$.
	The SPI does not evaluate this first form;
	rather, it solves Eq. \eqref{Eq:SEE2}, the second form of the SEE, which has been shown to be equivalent to Eq. \eqref{Eq:DerivedSEE} (a comprehensive proof can be found in the Supplement Material to Ref. \cite{SV13}).


\begin{thebibliography}{99}
\bibitem{comment:complexity}
		This can be seen from the observation that the SEEs for an operator of the form $\hat L=\hat L^{(1)}\otimes\cdots\otimes\hat L^{(N)}$ are identical to the EE for each $\hat L^{(j)}$ \cite{SV09,SV13}.
\bibitem{comment:positivity}
		A counterexample is replacing the positive operator in Eq. \eqref{Eq:ExampleRunOperator} with an only positive semidefinite one, $\hat L=\hat 1-\hat V$.
		This gives a non-converging (alternating) sequence of iterated vectors.
\end{thebibliography}
\end{document}